\documentclass{IOS-Book-Article}

\usepackage{mathptmx}
\usepackage{soul}\setuldepth{article}
\usepackage{amssymb}
\usepackage{amsmath}
\usepackage{amsfonts}
\usepackage{amsthm}
\usepackage{caption}
\usepackage[normalem]{ulem}
\usepackage[mathscr]{eucal}

\usepackage{hyperref}

\usepackage{tikz}
\usetikzlibrary{patterns,fit}   
\usepackage{subcaption}    

\newtheorem{theorem}{Theorem}
\newtheorem{definition}{Definition}
\newtheorem{proposition}{Proposition}

\newtheorem{example}{Example}

\usepackage{titlesec}

\def\hb{\hbox to 11.5 cm{}}

\begin{document}

\pagestyle{headings}
\def\thepage{}
\begin{frontmatter}              

\title{Beyond But-for Test: Counterfactual Explanation in Abstract Argumentation via Actual Causality (Extended Version)}


\author[A]{\fnms{Siyi} \snm{Liu}},
\author[A]{\fnms{Muyun} \snm{Shao}}
and
\author[A]{\fnms{Beishui} \snm{Liao}\thanks{Corresponding Author: Beishui Liao, baiseliao@zju.edu.cn}}

\address[A]{Zhejiang University, Hangzhou, China}

\begin{abstract}
Counterfactual explanation in abstract argumentation calls for an answer to the \emph{what-if} query: would the topic argument still be accepted if the status of certain other arguments were changed?
Existing approaches are limited to the but-for test and fail to accommodate more refined counterfactual conditions.
To overcome these limitations, we introduce an intervention-based counterfactual reasoning framework in abstract argumentation.
Our approach encodes the acceptance conditions of arguments as equations, then defines an intervention operator that supports (1) changing sets of arguments \emph{simultaneously}, and (2) fixing \emph{witness} arguments to their actual labels.
Guided by the refined counterfactual condition introduced in the Halpern-Pearl definition, our method goes beyond the but-for test, thereby correctly identifying causes in argumentation structures such as Preemption and Overdetermination.
Through comparison, we show that our method surpasses prior methods in both expressiveness and reliability.
\end{abstract}

\end{frontmatter}

\section{Introduction}
\label{sec:intro}
In the field of eXplainable AI (XAI), \emph{abstract argumentation} (AA) has been shown to have advantages in providing structured and verifiable explanations~\cite{cyrasArgumentativeXAI2021,engelmannsystematicreview2022,scheffersempiricalstudy2024}.
AA formalises conflicts within an argumentation framework (AF) and selects acceptable arguments via semantics~\cite{dungAcceptability1995}, which naturally supports answering the post-hoc explanatory question~\cite{amgoud2024posthoc}: \emph{why is an argument accepted under specific semantics?}
Existing argument-based explanation methods answer this question by identifying sets of arguments that satisfy formal properties, such as Relevance~\cite{fanComputingExplanations2015,liaoExplanationSemantics2020}, Monotonicity~\cite{ulbrichtStrongExplanations2021}, Sufficiency and Necessity~\cite{borgMNS2024,borgBasicFramework2021,borgNSExpl2021}.

In everyday reasoning, however, individuals rarely rely on static lists of reasons; they instinctively engage in counterfactual thinking, asking ``\emph{what-if}'' questions to probe causes and dependencies~\cite{millerExplanation2019,byrne2019counterfactuals}. In AA terms, this becomes a direct inquiry:
\begin{quote}
    \textit{Would the topic argument remain accepted if we were to change the acceptance status of one or more arguments suspected to be its causes?}
\end{quote}

This calls for a counterfactual reasoning account in AA.
Early work addresses this via the most primitive form of counterfactual reasoning, the \emph{but-for test}, which conducts a counterfactual analysis by asking whether the effect would have occurred without the cause.
In AA, this analysis is implemented by modifying the AF~\cite{SakamaCounterfactualReasoning2014,rienstra2014phdthesisarinflux,sakamaAbduction2018}.
Sakama~\cite{SakamaCounterfactualReasoning2014}, for instance, defines two primitive updates: adding a new initial argument that attacks an accepted argument to render it rejected; and deleting all incoming attacks on a rejected argument to render it accepted. In essence, these operations ask whether changing a \emph{single} argument's status would change the acceptance of the topic argument. In the resulting counterfactual, only the \emph{suspected cause} is altered, with no mechanism to hold other arguments fixed as \emph{witnesses}. Consider Preemption (Fig.~\ref{fig:preemption}), which illustrates the difficulty that arises when witnesses cannot be held fixed.

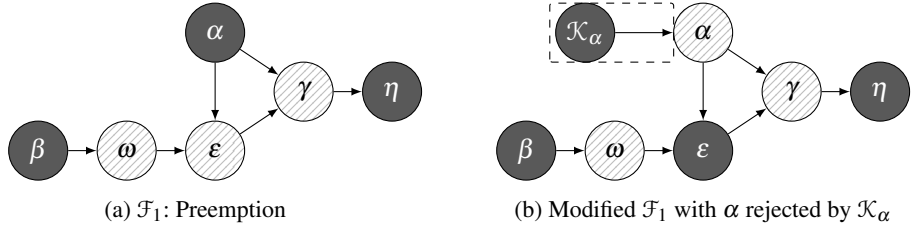
\begin{figure*}[htbp]
    \centering
    \begin{subfigure}[b]{0.4\textwidth}
        \centering
        \begin{tikzpicture}[scale=0.78,->,>=latex,auto,
        every node/.style={draw=black,circle,minimum size=0.78cm,inner sep=0pt},
        outnode/.style={text=black, pattern=north east lines, pattern color=gray!50},
        innode/.style={fill=gray!130, text=white},
        undnode/.style={fill=gray!50, text=white}
        ]
            \node[outnode] (a) at (0,-1) {$\epsilon$};
            \node[outnode] (b) at (-1.5,-1) {$\omega$};
            \node[innode] (c) at (-3,-1) {$\beta$};
            \node[innode] (d) at (0,1) {$\alpha$};
            \node[outnode] (e) at (1.5,0) {$\gamma$};
            \node[innode] (f) at (3,0) {$\eta$};
            \draw (c) -- (b); 
            \draw (b) -- (a);
            \draw (d) -- (a);
            \draw (a) -- (e); 
            \draw (d) -- (e); 
            \draw (e) -- (f);
        \end{tikzpicture}
        \caption{$\mathscr{F}_{1}$: Preemption}
        \label{fig:preemption}
    \end{subfigure}
    \hfill
    \begin{subfigure}[b]{0.5\textwidth}
        \centering
        \begin{tikzpicture}[scale=0.78,->,>=latex,auto,
        every node/.style={draw=black,circle,minimum size=0.78cm,inner sep=0pt},
        outnode/.style={text=black, pattern=north east lines, pattern color=gray!50},
        innode/.style={fill=gray!130, text=white},
        undnode/.style={fill=gray!50, text=white}
        ]
            \node[innode] (a) at (0,-1) {$\epsilon$};
            \node[outnode] (b) at (-1.5,-1) {$\omega$};
            \node[innode] (c) at (-3,-1) {$\beta$};
            \node[outnode] (d) at (0,1) {$\alpha$};
            \node[outnode] (e) at (1.5,0) {$\gamma$};
            \node[innode] (f) at (3,0) {$\eta$};
            \node[innode] (g) at (-2,1) {$\mathscr{K}_{\alpha}$};
            
            \draw (c) -- (b); 
            \draw (b) -- (a);
            \draw (d) -- (a);
            \draw (a) -- (e); 
            \draw (d) -- (e); 
            \draw (e) -- (f);
            \draw (g) -- (d);
           \draw[dashed, rounded corners=1pt]
            (-2.6, 0.5) rectangle (-0.5, 1.5);
        \end{tikzpicture}
        \caption{Modified $\mathscr{F}_1$ with $\alpha$ rejected by $\mathscr{K}_{\alpha}$}
        \label{fig:modifiedpreemption}
    \end{subfigure}
    \vspace{-3pt} 
    \caption{Counterexample to the but-for test in AA. The status of each argument is visually encoded as: $\mathbf{in}$=dark gray (accepted), $\mathbf{out}$=cross-hatched (rejected), $\mathbf{und}$=light gray (undecided). The dashed box encloses the newly introduced argument and the attack relation. These visual conventions are used throughout the paper.}
    \label{fig:two-structures}
    \vspace{-8pt} 
\end{figure*}

\begin{example}\label{ex:preemption}
    Consider the AF $\mathscr{F}_{1}$ in Fig.~\ref{fig:preemption}, where $\eta$ is the accepted argument to be explained. $\alpha$ and $\beta$ are initial and thus accepted. From each, an even-length path reaches $\eta$. However, only the ``protecting'' path from $\alpha$ is active, since $\alpha$ also attacks $\epsilon$, keeping it rejected and thereby blocking the defensive path from $\beta$ to $\eta$.
    
    Intuitively, $\alpha$ contributes to $\eta$'s acceptance, whereas $\beta$ does not, leading to the conclusion that $\alpha$ should be qualified as a cause of $\eta$'s acceptance. However, rejecting $\alpha$ (by adding an attacker) does not alter $\eta$'s acceptance: with $\alpha$ out, it no longer keeps $\epsilon$ rejected, reactivating the defensive path from $\beta$ via $\omega$--$\epsilon$--$\gamma$ to $\eta$ (Fig.~\ref{fig:modifiedpreemption}). The problem is that the counterfactual scenario does not keep $\epsilon$ fixed to its actual rejected status, allowing the alternative pathway to distort the counterfactual evaluation.
\end{example}

To address the limitations of the but-for test highlighted above, we build on the expressive framework of structural equation models (SEMs)~\cite{pearl2009causality} and the Halpern–Pearl (HP) definition of actual causality~\cite{HalpernandPearl2005,halpern2015modification,halpern2016actual}.
We encode argument acceptance via equations and define an intervention operator that supports manipulating \emph{sets} of arguments while keeping \emph{witness} arguments fixed. 
Building on this, our approach yields a refined counterfactual condition that goes beyond the but‑for test and properly handles cases where the but‑for test fails, namely Preemption and Overdetermination.
A detailed analysis of these cases is provided in Section~\ref{sec:counterexpl}.

The structure of this paper is as follows.
Section~\ref{sec:pre} recalls the necessary background.
Section~\ref{sec:model} introduces the acceptability model and the intervention operator.
Section~\ref{sec:counterexpl} defines the notion of actual cause via the refined counterfactual condition $\mathrm{AC2}^m$ and forms the basis for the counterfactual explanation method advocated in this paper.
Section~\ref{sec:IntandMod} presents the graphical representation for visualization.
Section~\ref{sec:relatedwork} compares prior methods.
Section~\ref{sec:con} concludes and outlines future work. All proofs of theorems and propositions are provided in the Appendix.

\section{Preliminaries}
\label{sec:pre}

An \emph{argumentation framework} (AF) is a directed graph $\mathscr{F}=(\mathscr{A},\mathscr{R})$, where $\mathscr{A}$ is a finite set of \emph{arguments} and $\mathscr{R}\subseteq \mathscr{A}\times \mathscr{A}$ represents the \emph{attack relation}~\cite{dungAcceptability1995}. The universe of AF is represented by the set $\mathscr{UF}$, while the set of arguments that appear within $\mathscr{UF}$ is denoted as $\mathscr{UA}$. For an AF $\mathscr{F}=(\mathscr{A},\mathscr{R})$ and $\alpha,\beta\in \mathscr{A}$, define $\alpha^{+}=\{\beta\in \mathscr{A}\mid (\alpha,\beta)\in \mathscr{R}\}$ and $\alpha^{-}=\{\beta\in \mathscr{A}\mid (\beta,\alpha)\in \mathscr{R}\}$.

Given an AF $\mathscr{F}=(\mathscr{A},\mathscr{R})$, a set $E\subseteq \mathscr{A}$ is \emph{conflict-free}, denoted as $E\in \mathscr{E}_{cf}(\mathscr{F})$ iff for all $\alpha,\beta\in E$, $(\alpha,\beta)\notin \mathscr{R}$. $E$ \emph{defends} an argument $\alpha\in \mathscr{A}$, iff for all $\beta\in \alpha^-$, there exists $\gamma\in E$, s.t. $(\gamma,\beta)\in \mathscr{R}$. 
The characteristic function of an AF $\mathscr{F}=(\mathscr{A},\mathscr{R})$ is a function $\Gamma$ s.t. for every $E\subseteq A$, we have $\Gamma_{\mathscr{F}}(E)=\{\alpha\in \mathscr{A}\mid E~\text{defends}~\alpha\}$.

To characterize how an argument \emph{(in)directly defends} another, we introduce the notion of directed and defense paths.

\begin{definition}\label{def:path-relevant}
    Let $\mathscr{F}=(\mathscr{A},\mathscr{R})$ be an AF, and $\alpha,\beta\in \mathscr{A}$. A sequence $P=(\alpha_0,\cdots,\alpha_n)$ is a directed path from $\alpha$ to $\beta$ iff $\alpha_0=\alpha$, $\alpha_n=\beta$, and $(\alpha_{i-1},\alpha_i)\in \mathscr{R}$ for all $i\in \{1,\cdots,n\}$; its length is $n$ (number of edges). If such a path exists, $\alpha$ is also said to be relevant to $\beta$ (by convention, each $\alpha$ is relevant to itself via a path of length $0$).
    For a set $S\subseteq \mathscr{A}$, $S$ is relevant for $\alpha$ if every $\beta\in S$ is relevant to $\alpha$.

    A directed path from $\alpha$ to $\beta$ is a defense path if its length is even. Furthermore, if there exists a directed path $P=(\alpha,\gamma,\beta)$ of length two, s.t. $\gamma$ is distinct from $\alpha$ and $\beta$, we say that $\alpha$ directly defends $\beta$.
\end{definition}

The acceptability of an argument is prescribed by specific \emph{argumentation semantics} ~\cite{baroniAbstractAF2018}. Formally, a semantics $\sigma$ assigns to each AF $\mathscr{F}=(\mathscr{A},\mathscr{R})$ a subset of $2^{\mathscr{A}}$, denoted as $\mathscr{E}_\sigma(\mathscr{F})$.
Classical semantics includes \emph{admissible}, \emph{complete}, \emph{grounded}, \emph{preferred}, and \emph{stable} semantics (abbr. $ad,co,gr,pr,st$). 

An extension under the semantics $\sigma$ is called a $\sigma$-extension.
While we focus on \emph{complete} ($co$) semantics as our baseline, we outline in Section.~\ref{sec:model} how the formal framework for counterfactual explanation can be adapted to other semantics.
Let $\mathscr{F}=(\mathscr{A},\mathscr{R})$ be an AF, and $E\in \mathscr{E}_{cf}(\mathscr{F})$, $E\in \mathscr{E}_{co}(\mathscr{F})$ iff $E$ defends all its elements ($E\in \mathscr{E}_{ad}(\mathscr{F})$), and $E=\Gamma_{\mathscr{F}}(E)$.

Argumentation semantics can be defined through two parallel methods: the above \emph{extension-based} method, and the \emph{labelling-based} method, which formulates semantics in terms of $\sigma$-labelling~\cite{caminadalogical2009}.

\begin{definition}\label{def:labelling}
    Given an AF $\mathscr{F}=(\mathscr{A},\mathscr{R})$, a labelling for $\mathscr{F}$ is a total function $\mathscr{L}: \mathscr{A}\mapsto \mathbb{V}$, that assigns each argument a label from the set $\mathbb{V}=\{\mathbf{in},\mathbf{out},\mathbf{und}\}$, denoted as $\mathscr{L}_{\sigma}(\mathscr{F})$ for $\sigma\in \{ad,co,gr,pr,st\}$. 
\end{definition}

Let $\mathbf{in}(\mathscr{L})$, $\mathbf{out}(\mathscr{L})$, and $\mathbf{und}(\mathscr{L})$ be the sets of arguments labelled $\mathbf{in}$, $\mathbf{out}$, and $\mathbf{und}$ respectively. We use the triple $\big\langle \mathbf{in}(\mathscr{L}),\mathbf{out}(\mathscr{L}),\mathbf{und}(\mathscr{L})\big\rangle$ to represent the labelling $\mathscr{L}$.

A labelling $\mathscr{L}$ of $\mathscr{F}=(\mathscr{A},\mathscr{R})$ is said to be \textit{admissible} ($\mathscr{L}\in \mathscr{L}_{ad}(\mathscr{F})$) iff $\forall \alpha\in \mathbf{in}(\mathscr{L})\cup \mathbf{out}(\mathscr{L})$ it holds that: (i) $\mathscr{L}(\alpha)=\mathbf{out}$ iff $\exists (\beta,\alpha)\in \mathscr{R}$ s.t. $\mathscr{L}(\beta)=\mathbf{in}$; and (ii) $\mathscr{L}(\alpha)=\mathbf{in}$ iff $\forall (\beta,\alpha)\in \mathscr{R}$, $\mathscr{L}(\beta)=\mathbf{out}$. Moreover, $\mathscr{L}$ is a \emph{complete labelling} ($\mathscr{L}\in \mathscr{L}_{co}(\mathscr{F})$) iff conditions (i) and (ii) hold for all arguments in $\mathscr{A}$.

\section{Intervention-based Counterfactual Reasoning in AA}
\label{sec:model}
\subsection{Acceptability Model and Intervention}

We start by defining a language of the logical system for reasoning about \emph{argument-label pairs}. The \emph{atomic propositions} are of the form $\mathbf{v}(\alpha)$, asserting that ``argument $\alpha$ has the label $\mathbf{v}$''. The set of logical connectives $\{\neg,\wedge\}$ forms an adequate set, where $\neg \mathbf{v}(\alpha)$ is interpreted as ``the label of $\alpha$ is not $\mathbf{v}$''.
The language further includes \emph{counterfactual formulas} $[\psi]\varphi$, where $[\psi]$ is an \emph{intervention operator},
read as ``after intervening to make $\psi$ true, $\varphi$ holds''. Here, $\psi$, called a \emph{hypothesis}, is a consistent conjunction of atomic propositions, ensured by the constraint that no argument receives two different labels.

\begin{definition}[Syntax]\label{def:syntax}
    Given a universal set of arguments $\mathscr{UA}$ and the range of label $\mathbb{V}$, with $\alpha$ ranging over $\mathscr{UA}$ and $\mathbf{v}$ ranging over $\mathbb{V}$, the syntax of the language $L_{\mathscr{UA}}$ is defined recursively as follows:
    \[
    \varphi::=\mathbf{v}(\alpha) \mid \neg\varphi\mid \varphi \wedge\varphi\mid [\psi]\varphi,
    \]
    where in $[\psi]\varphi$, the hypothesis $\psi$ is defined by the grammar: $\psi ::= \top\mid \mathbf{v}(\alpha) \mid \psi \land \psi$, with the constraint that for any distinct atomic subformula $\mathbf{v_1}(\alpha_1),\mathbf{v_2}(\alpha_2)\in \mathsf{Sub}(\psi)$, $\alpha_{1}\neq \alpha_2$, where $\mathsf{Sub}(\psi)$ denotes the set of all atomic subformulas of $\psi$.
\end{definition}

A hypothesis $\psi$ is called an \emph{atomic hypothesis} if $\psi=\mathbf{v}(\alpha)$ for $\alpha\in \mathscr{UA}$ and $\mathbf{v}\in \mathbb{V}$; otherwise, it is a \emph{compound hypothesis}. And for all $\mathbf{v_i}(\alpha_i)\in \mathsf{Sub}(\psi)$, denote the set of $\alpha_{i}$ by $\mathscr{A}_{\psi}$.
We write $L_{\mathscr{UA}}^{-}$ for the fragment of $L_{\mathscr{UA}}$ without counterfactual formulas.

Given an AF $\mathscr{F}$, there exists a \emph{unique} \emph{acceptability model}, denoted as $\mathscr{M}_{\mathscr{F}}$, which is a set of \emph{acceptability equations}. For any argument $\alpha$ in $\mathscr{F}$, its equation $f_{\alpha}$ is a function from the labels of $\alpha$'s attackers to the label of $\alpha$.

\begin{definition}[Acceptability Model]\label{def:acceptabilityModel}
    Given an AF $\mathscr{F}=(\mathscr{A},\mathscr{R})$, its acceptability model $\mathscr{M}_{\mathscr{F}}=\{f_{\alpha}\}_{\alpha\in \mathscr{A}}$ is the unique set of functions, where each $f_{\alpha}$ is called an acceptability equation for $\alpha$, defined as follows.

    For each $\alpha\in \mathscr{A}$. The function $f_{\alpha}$ takes as input an assignment of labels to each attacker $\beta\in\alpha^{-}$, denoted by $(\mathbf{v}_{\beta})_{\beta\in \alpha^{-}}$ with $\mathbf{v}_{\beta}\in\mathbb{V}$, and returns a label in $\mathbb{V}$ as follows:
    \[
    f_{\alpha}\big( (\mathbf{v}_{\beta})_{\beta\in \alpha^{-}} \big) = 
    \begin{cases}
    \mathbf{in} & \text{if } \forall \beta\in \alpha^{-}: \mathbf{v}_{\beta} = \mathbf{out}, \\[1pt]
    \mathbf{out} & \text{if } \exists \beta\in \alpha^{-}: \mathbf{v}_{\beta} = \mathbf{in}, \\[1pt]
    \mathbf{und} & \text{otherwise}.
    \end{cases}
    \]
    
\end{definition}

Our framework adopts the \emph{intervention‑based semantics} of SEMs~\cite{pearl2009causality} to define counterfactual reasoning.
An intervention modifies the model $\mathscr{M}_{\mathscr{F}}$ according to a hypothesis $\psi$, producing an \emph{intervened model} $\mathscr{M}_{\mathscr{F}}^{\psi}$ where the acceptability equation of each argument appearing in $\psi$ is replaced by a constant function returning its stipulated label, while all other equations remain unchanged.

\begin{definition}[Intervention]\label{def:Intervention}
    Given an AF $\mathscr{F}=(\mathscr{A},\mathscr{R})$ and its acceptability model $\mathscr{M}_{\mathscr{F}}=\{f_{\alpha}\}_{\alpha\in\mathscr{A}}$, let $\psi=\bigwedge_{i=1}^{k}\mathbf{v_i}(\alpha_i)$ (with $\alpha_i\in\mathscr{A}$ and $\mathbf{v_i}\in\mathbb{V}$) be a hypothesis. The intervened model $\mathscr{M}_{\mathscr{F}}^{\psi}=\{f_{\alpha}^{\psi}\}_{\alpha\in\mathscr{A}}$ is defined by: for each $\alpha\in \mathscr{A}$, a function $f_{\alpha}^{\psi}$ on the same domain as $f_{\alpha}$:
    \[
    f_{\alpha}^{\psi} = 
    \begin{cases}
        \mathbf{v_i}, & \exists i\in \{1,\cdots,k\}: \alpha=\alpha_i,  \\[1pt]
        f_{\alpha} & \text{otherwise}.
    \end{cases}
    \]
\end{definition}

Recalled from Def.~\ref{def:labelling}, a labelling $\mathscr{L}$ is a total function that assigns a label to each argument in an AF $\mathscr{F}$. To unify non-intervened and intervened models, we have introduced $\top$ and let $\mathscr{M}_{\mathscr{F}}=\mathscr{M}_{\mathscr{F}}^{\top}$.
A \emph{solution} of an (intervened) model $\mathscr{M}_\mathscr{F}$ is a labelling that satisfies all its equations. Unlike classic SEMs, which are typically acyclic and yield a unique solution once the exogenous variables are fixed, an AF may contain cycles and thus admit multiple solutions, which correspond to the multiple $co$-labellings.


\begin{definition}[Solution]\label{def:solution}
    Given an AF $\mathscr{F}=(\mathscr{A},\mathscr{R})$, and its acceptability model $\mathscr{M}_{\mathscr{F}}$.
    A solution of the (intervened) model $\mathscr{M}_{\mathscr{F}}^{\psi}$ is a labelling $\mathscr{L}: \mathscr{A} \mapsto \mathbb{V}$, s.t. for all $\alpha\in\mathscr{A}$:
    \[
    \mathscr{L}(\alpha) = f_{\alpha}^{\psi}\big( (\mathscr{L}(\beta))_{\beta\in \alpha^{-}} \big),~\text{where}~f_{\alpha}^{\psi}\in \mathscr{M}_{\mathscr{F}}^{\psi}.
    \]
\end{definition}

\begin{example}\label{ex:runningexample}
    Consider the AF $\mathscr{F}_{2}$ in Fig.~\ref{fig:runningexample}. 
    For $\alpha$ whose only attacker is $\gamma\in \alpha^{-}$, the acceptability equation $f_{\alpha}$ can be expressed as:    
    \begin{center}
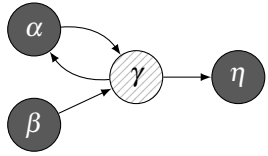
        \begin{minipage}[c]{0.4\textwidth}
            \centering
            \begin{tikzpicture}[->,>=latex,auto,
                every node/.style={draw=black,circle,minimum size=0.7cm,inner sep=0pt},
                outnode/.style={text=black, pattern=north east lines, pattern color=gray!50},
                innode/.style={fill=gray!130, text=white},
                undnode/.style={fill=gray!50, text=white},
                scale=0.9
            ]
                \node[outnode] (c) at (0,0) {$\gamma$};
                \node[innode] (a) at (-1.5,0.7) {$\alpha$};
                \node[innode] (b) at (-1.5,-0.7) {$\beta$};
                \node[innode] (d) at (1.5,0) {$\eta$};
                \draw (b)--(c);
                \draw (c)--(d);
                \draw (a) to[bend left] (c);
                \draw (c) to[bend left] (a);
            \end{tikzpicture}
            \captionof{figure}{$\mathscr{F}_2$ with Even-Cycle}
            \label{fig:runningexample}
        \end{minipage}
        \hfill
        \begin{minipage}[c]{0.55\textwidth}
        \vspace{-10pt}
            \[
            f_{\alpha}=\big\{\, \big((\mathbf{in}),\mathbf{out}\big),\;
            \big((\mathbf{out}),\mathbf{in}\big),\;
            \big((\mathbf{und}),\mathbf{und}\big) \,\big\},
            \]
            where each pair $\big((\mathbf{v}_{\gamma}),\mathbf{v}_{\alpha}\big)$ denotes that input label $\mathbf{v}_{\gamma}$ yields output label $\mathbf{v}_{\alpha}$.
        \end{minipage}
    \end{center}
    \vspace{-3pt}
    
    Intervening with the hypothesis $\psi=\mathbf{out}(\alpha)$, the equation $f_{\alpha}$ is replaced by the constant function $f_{\alpha}^{\psi}$ that always returns $\mathbf{out}$. The unique solution of the intervened model $\mathscr{M}_{\mathscr{F}_{2}}^{\psi}$ is the labelling $\mathscr{L}=\big\langle \{\beta,\eta\},\{\alpha,\gamma\},\emptyset \big\rangle$.
\end{example}

The set of all labellings that are solutions of $\mathscr{M}_{\mathscr{F}}^{\psi}$ is denoted as $\mathsf{Sol}(\mathscr{M}_{\mathscr{F}}^{\psi})$. Without intervention, the solutions of the model $\mathscr{M}_{\mathscr{F}}$ coincide with the $co$-labellings of the AF.

\begin{theorem}\label{the:SolandL}
    For any AF $\mathscr{F}=(\mathscr{A},\mathscr{R})$, and its acceptability model $\mathscr{M}_{\mathscr{F}}$, it holds that $\mathsf{Sol}(\mathscr{M}_{\mathscr{F}})=\mathscr{L}_{co}(\mathscr{F})$.
\end{theorem}

The solution of a model $\mathscr{M}_{\mathscr{F}}^{\psi}$ (Def.~\ref{def:solution}) is tied to \emph{complete} semantics but can extend to others (e.g., $gr,pr,st$) by constraining the solution set $\mathsf{Sol}(\mathscr{M}_{\mathscr{F}}^{\psi})$. For \emph{grounded} semantics, the solution is restricted to the one where $\mathbf{in}(\mathscr{L})$ is minimal w.r.t. set-inclusion. In Ex.~\ref{ex:runningexample}, the intervened model $\mathscr{M}_{\mathscr{F}_{2}}^{\psi}$ admits a unique solution $\mathscr{L}$, which strictly satisfies set-inclusion minimality and thus represents the exact grounded solution.

\subsection{Semantics: Base and Model-level Satisfaction}

Next, we define the semantics of this logical system. We introduce two levels of satisfaction: \emph{base satisfaction} and \emph{model-level satisfaction}. The base satisfaction relation $(\mathscr{M},\mathscr{L})\Vdash \varphi$ evaluates a formula $\varphi\in L_{\mathscr{UA}}^{-}$ which is not a counterfactual formula, w.r.t. a specific solution $\mathscr{L}$ of a (possibly intervened) model $\mathscr{M}$. 

\begin{definition}[Base Satisfaction]
     Given an AF $\mathscr{F}=(\mathscr{A},\mathscr{R})$ and its acceptability model $\mathscr{M}_{\mathscr{F}}$. Consider any (intervened) model $\mathscr{M}=\mathscr{M}_{\mathscr{F}}^{\psi}$ and a solution $\mathscr{L}$ of $\mathscr{M}$, the base satisfaction relation $(\mathscr{M},\mathscr{L})\Vdash \varphi$ for formula $\varphi\in L_{\mathscr{UA}}^{-}$ is defined recursively:
     \begin{align*}
     (\mathscr{M},\mathscr{L}) \Vdash \mathbf{v}(\alpha) &\iff \mathscr{L}(\alpha) = \mathbf{v}, \\
     (\mathscr{M}, \mathscr{L}) \Vdash \neg\varphi &\iff (\mathscr{M}, \mathscr{L}) \nVdash \varphi, \\
     (\mathscr{M}, \mathscr{L}) \Vdash \varphi \land \psi &\iff 
        (\mathscr{M}, \mathscr{L}) \Vdash \varphi \text{ and } (\mathscr{M}, \mathscr{L}) \Vdash \psi.
    \end{align*}
\end{definition}

Model-level satisfaction lifts the base evaluation, considering either \emph{all} ($\models_{\forall}$) or \emph{some} ($\models_{\exists}$) of its solutions, and defines the meaning of counterfactual formulas.

\begin{definition}[Model-level Satisfaction]\label{def:Model-Satisfaction}
    Let $\varphi \in L_{\mathscr{UA}}$ be a formula, $\mathscr{F} = (\mathscr{A}, \mathscr{R})$ an AF, and $\mathscr{M} = \mathscr{M}_{\mathscr{F}}^{\psi}$ any (intervened) acceptability model. \textbf{Strong satisfaction ($\models_{\forall}$)} and \textbf{weak satisfaction ($\models_{\exists}$)} are defined recursively over the structure of $\varphi$:
    \begin{align*}
    \text{For } \varphi \in L_{\mathscr{UA}}^{-}:\quad
    &\mathscr{M} \models_{\forall} \varphi \iff \forall \mathscr{L} \in \mathsf{Sol}(\mathscr{M}):\, (\mathscr{M}, \mathscr{L}) \Vdash \varphi,\\
    &\mathscr{M} \models_{\exists} \varphi \iff \exists \mathscr{L} \in \mathsf{Sol}(\mathscr{M}):\, (\mathscr{M}, \mathscr{L}) \Vdash \varphi.\\[4pt]
    \text{For } \varphi = [\psi]\chi \text{ with } \psi,\chi \in L_{\mathscr{UA}}^{-}:\quad
    &\mathscr{M} \models_{\forall} [\psi]\chi \iff \mathscr{M}^{\psi} \models_{\forall} \chi,\\
    &\mathscr{M} \models_{\exists} [\psi]\chi \iff \mathscr{M}^{\psi} \models_{\exists} \chi.
    \end{align*}

\end{definition}

\begin{example}[Cont. of Ex.~\ref{ex:runningexample}]
    Recall $\mathscr{F}_{2}$ in Fig.~\ref{fig:runningexample} and its unique $co$-labelling $\mathscr{L}\in \mathscr{L}_{co}(\mathscr{F}_2)$ depicted therein. Base satisfaction gives $(\mathscr{M}_{\mathscr{F}_2},\mathscr{L})\Vdash \mathbf{in}(\alpha)$. Consider two interventions: $\psi_1=\mathbf{out}(\alpha)$ and $\psi_2=\mathbf{out}(\beta)$.
    \begin{itemize}
        \item Under $\psi_1=\mathbf{out}(\alpha)$, $\mathscr{M}_{\mathscr{F}_2}^{\psi_1}$ has the unique solution $\mathscr{L}_{1}=\big\langle \{\beta,\eta\},\{\alpha,\gamma\},\emptyset\big\rangle$, where $\eta$ labelled $\mathbf{in}$. Thus, $\mathscr{M}_{\mathscr{F}_{2}}\models_{\forall} [\mathbf{out}(\alpha)]\mathbf{in}(\eta)\iff \mathscr{M}_{\mathscr{F}_{2}}^{\mathbf{out}(\alpha)} \models_{\forall} \mathbf{in}(\eta)$.
        \item Under $\psi_2=\mathbf{out}(\beta)$, one of the solutions of $\mathscr{M}_{\mathscr{F}_2}^{\psi_2}$ is $\mathscr{L}_2 = \big\langle \emptyset, \{\beta\}, \{\alpha, \gamma, \eta\} \big\rangle$, where $\eta$ labelled $\mathbf{und}$. Thus, $\mathscr{M}_{\mathscr{F}_{2}}\models_{\exists} [\mathbf{out}(\beta)]\mathbf{out}(\eta)\iff \mathscr{M}_{\mathscr{F}_{2}}^{\mathbf{out}(\beta)} \models_{\exists} \mathbf{out}(\eta)$.
    \end{itemize}

\end{example}

\section{Counterfactual Explanation based on Actual Causality}\label{sec:counterexpl}

In general, an \emph{explanation strategy} maps an AF $\mathscr{F}$ and a topic argument $\alpha$ in $\mathscr{F}$ to a set of explanations of a specific type. Here, we focus on the \emph{formula-based} strategy. A formula‑based explanation method is a function $\operatorname{F-Expl}: \mathscr{UF} \times \mathscr{UA} \to 2^{L_{\mathscr{UA}}^{-}}$ associating with each AF $\mathscr{F}=(\mathscr{A},\mathscr{R})$ and $\alpha\in \mathscr{A}$ a subset of $L_{\mathscr{UA}}^{-}$, denoted as $\operatorname{F-Expl}(\mathscr{F},\alpha)$.
Each $\psi\in\operatorname{F-Expl}(\mathscr{F},\alpha)$ is called a \emph{formula‑based explanation}.

\subsection{From But-for Cause to Actual Causes}
\label{subsec:butfortoactual}

To formalise ``changing the cause(s)'', we introduce the set of \emph{modified hypotheses} for the hypothesis $\psi$: those obtained by changing the label of at least one argument in $\psi$ to a different label while keeping the other arguments fixed.

\begin{definition}[Modified Hypotheses]\label{def:modified-hyp}
    Let $\psi = \bigwedge_{i=1}^k \mathbf{v_i}(\alpha_i)$ be a hypothesis. The set of modified hypotheses of $\psi$, denoted by $\mathsf{Mod}(\psi)$, comprises all $\psi' = \bigwedge_{i=1}^k \mathbf{v_i^{\prime}}(\alpha_i)$ s.t. there exists at least one index $i \in \{1,\dots,k\}$ with $\mathbf{v_i^{\prime}} \neq \mathbf{v_i}$.
\end{definition}

We consider two counterfactual conditions: the but‑for test $\mathrm{BFT}$ (underlying graph‑modification, Ex.~\ref{ex:preemption}) and the modified condition $\mathrm{AC2}^m$.
A hypothesis passing $\mathrm{BFT}$ is called a but‑for cause, and one 
passing $\mathrm{AC2}^m$ an actual cause.
Both admit strong ($\forall$) and weak ($\exists$) versions, as an AF 
may admit multiple $co$-labellings due to cycles.

The \emph{but-for} test ($\mathrm{BFT}$) applies only to an atomic hypothesis of the form $\mathbf{v}(\alpha)$ and involves no witness. It asks whether the effect would have occurred without the cause and, in the AA setting, modifies the suspected argument's label to either $\mathbf{in}$ or $\mathbf{out}$. We define the restricted modification set $\mathsf{Mod}_{\mathbf{in/out}}(\psi)\subseteq\mathsf{Mod}(\psi)$, which excludes modifications to $\mathbf{und}$, following the graph‑modification approach of Sakama~\cite{SakamaCounterfactualReasoning2014}.

\begin{definition}[But-for Cause]\label{def:butforcause}
    Let $\mathscr{F}=(\mathscr{A},\mathscr{R})$ be an AF, $\mathscr{M}_{\mathscr{F}}$ its acceptability model, $\alpha\in\mathscr{A}$ the topic, and $\mathscr{L}$ a solution of $\mathscr{M}_{\mathscr{F}}$.
    An atomic hypothesis $\psi\in L_{\mathscr{UA}}^{-}$ is a but-for cause for $\mathbf{v}(\alpha)$ with $(\mathscr{M},\mathscr{L})\Vdash \psi\wedge \mathbf{v}(\alpha)$, iff it passes the but-for test $\mathrm{BFT}_{q}$, formally:
    \[
    \exists \psi^{\prime}\in \mathsf{Mod}_{\mathbf{in/out}}(\psi): \mathscr{M}_{\mathscr{F}}\models_{q}[\psi^{\prime}]\neg \mathbf{v}(\alpha) \text{ for } q\in \{\forall,\exists\}.
    \]
\end{definition}

Our definition of an actual cause in AA preserves the three HP-conditions of actual causality~\cite[Def.~2.1]{halpern2015modification}:
Actuality requires the cause and effect to hold in the original labelling; Counterfactual Condition uses the modified version $\mathrm{AC2}^m$ adopted in this paper;
\footnote{In the original HP-definition, this condition is denoted by $\mathrm{AC2(a^m)}$~\cite{halpern2015modification}.}
Minimality ensures that the cause contains no redundant parts.

\begin{definition}[Actual Cause]\label{def:actualcauses}
    Let $\mathscr{F}=(\mathscr{A},\mathscr{R})$ be an AF and $\mathscr{M}_{\mathscr{F}}$ its acceptability model. Consider $\alpha\in \mathscr{A}$ the topic, and $\mathscr{L}$ a solution of $\mathscr{M}_{\mathscr{F}}$.
    A hypothesis $\psi\in L_{\mathscr{UA}}^{-}$ is an actual cause for $\mathbf{v}(\alpha)$ iff it satisfies the following three conditions:
    \begin{description}
        \item[Actuality ($\mathrm{AC1}$):] $(\mathscr{M}_{\mathscr{F}},\mathscr{L})\Vdash \psi\wedge \mathbf{v}(\alpha)$.
        \item[Modified Counterfactual Condition ($\mathrm{AC2}^m$):] In Def.~\ref{def:ac2(am)} below.
        \item[Minimality ($\mathrm{AC3}$):] Let $\psi=\bigwedge_{i=1}^k \mathbf{v_i}(\alpha_i)$. Then for every proper sub-conjunction $\psi_J = \bigwedge_{i\in J} \mathbf{v_i}(\alpha_i)$ ($J \subsetneq \{1,\dots,k\}$), $\psi_{J}$ fails to satisfy both $\mathrm{AC1}$ and $\mathrm{AC2}$.
    \end{description}

\end{definition}

From Def.~\ref{def:butforcause}, a but-for cause trivially satisfies $\mathrm{AC1}$ and $\mathrm{AC3}$. The distinction between a but-for cause and an actual cause lies solely in the counterfactual condition.

\begin{proposition}\label{pro:forbut-for}
     Let $\mathscr{M}_{\mathscr{F}}$ be an acceptability model with a topic argument $\alpha$.
     If hypothesis $\psi$ is a but-for cause for $\mathbf{v}(\alpha)$, then $\psi$ satisfies Actuality ($\mathrm{AC1}$) and Minimality ($\mathrm{AC3}$).
\end{proposition}

$\mathrm{AC2}^m$ improves upon $\mathrm{BFT}$ by allowing \emph{compound hypothesis} and keeping \emph{witnesses} fixed to their actual labels. It validates a counterfactual if there \emph{exist} fixed witnesses that enable the intervention to alter the topic's label. Analyzing the specific topological properties of such witnesses is left for future work.

\begin{definition}[Modified Counterfactual Condition ($\mathrm{AC2}^m$)]\label{def:ac2(am)}
     Let $\psi$ and $\chi$ be two hypotheses with $\mathscr{A}_{\psi}\cap \mathscr{A}_{\chi}=\emptyset$, $\alpha\notin\mathscr{A}_{\chi}$, and $(\mathscr{M}_{\mathscr{F}},\mathscr{L})\Vdash \chi$. We say $\psi$ satisfies $\mathrm{AC2}^m_q$ iff:
     \[
    \exists \psi^{\prime}\in \mathsf{Mod}(\psi): \mathscr{M}_{\mathscr{F}}\models_{q}[\psi^{\prime}\wedge \chi]\neg \mathbf{v}(\alpha)\text{~for~}q\in \{\forall,\exists\}.
     \]
\end{definition}

\begin{example}[But-for vs. Actual Cause]\label{ex:butforvs.actualcause}
    Recall the AF $\mathscr{F}_{1}$ in Fig.~\ref{fig:preemption} and its $co$-labelling $\mathscr{L}$ depicted therein, where the topic argument $\eta$ is labelled $\mathbf{in}$. Let the suspected cause of $\mathbf{in}(\eta)$ be the hypothesis $\psi=\mathbf{in}(\alpha)$.

    $\mathbf{in}(\alpha)$ fails the but-for test $\mathrm{BFT}_{\exists}$: after intervening with $\mathbf{out}(\alpha)$, $\mathscr{M}_{\mathscr{F}_1}^{\mathbf{out}(\alpha)}$ has the unique solution $\mathscr{L}_1 = \big\langle \{\beta,\epsilon,\eta\}, \{\alpha, \omega, \gamma\}, \emptyset \big\rangle$, in which $\eta$ remains $\mathbf{in}$. Hence, $\neg\mathbf{in}(\eta)$ does not hold, and $\mathbf{in}(\alpha)$ does not qualify as a but‑for cause.

    Yet $\mathbf{in}(\alpha)$ passes $\mathrm{AC2}^m_{\exists}$ with witness $\chi=\mathbf{out}(\epsilon)$.
    Specifically, there exists a witness $\chi=\mathbf{out}(\epsilon)$ fixed to its actual label in $\mathscr{L}$, s.t. once we intervene with $\psi_1=\mathbf{out}(\alpha) \in \mathsf{Mod}(\psi)$ while holding $\chi$ fixed, the compound hypothesis $\psi_2 = \mathbf{out}(\alpha) \land \mathbf{out}(\epsilon)$ is formed.
    The intervened model $\mathscr{M}_{\mathscr{F}_1}^{\psi_2}$ then yields the unique solution $\mathscr{L}_2 = \big\langle \{\beta,\gamma\}, \{\alpha, \omega, \epsilon, \eta\}, \emptyset \big\rangle$, in which $\eta$ is labelled $\mathbf{out}$.
    Thus $\mathbf{in}(\alpha)$ is identified as an actual cause of $\mathbf{in}(\eta)$, overcoming the counter-intuition of the but-for test illustrated in Ex.~\ref{ex:preemption}.
\end{example}

Theorem~\ref{theorem:StrongtoWeak} establishes that for both $\mathrm{BFT}$ and 
$\mathrm{AC2}^m$ the strong variant implies the weak variant, while Theorem~\ref{theorem:threevariants} shows that every but-for cause is also an actual cause. Together, they characterize the logical relationships among the counterfactual conditions.

\begin{theorem}\label{theorem:StrongtoWeak}
    Let $\mathscr{M}_{\mathscr{F}}$ be an acceptability model with a topic argument $\alpha$. For any atomic hypothesis $\psi \in L_{\mathscr{UA}}^{-}$, if $\psi$ is a but-for cause for $\mathbf{v}(\alpha)$ under $\mathrm{BFT}_{\forall}$, then it is also a but-for cause for $\mathbf{v}(\alpha)$ under $\mathrm{BFT}_{\exists}$. For any hypothesis $\psi \in L_{\mathscr{UA}}^{-}$, if $\psi$ is an actual cause for $\mathbf{v}(\alpha)$ under $\mathrm{AC2}^m_{\forall}$, then it is also an actual cause for $\mathbf{v}(\alpha)$ under $\mathrm{AC2}^m_{\exists}$.

\end{theorem}

\begin{theorem}\label{theorem:threevariants}
    Let $\mathscr{M}_{\mathscr{F}}$ be an acceptability model with a topic argument $\alpha$. For any atomic hypothesis $\psi\in L_{\mathscr{UA}}^{-}$ and $q\in \{\forall,\exists\}$, if $\psi$ is a but-for cause for $\mathbf{v}(\alpha)$ under $\mathrm{BFT}_q$, then it is also an actual cause for $\mathbf{v}(\alpha)$ under $\mathrm{AC2}^m_{q}$.
\end{theorem}

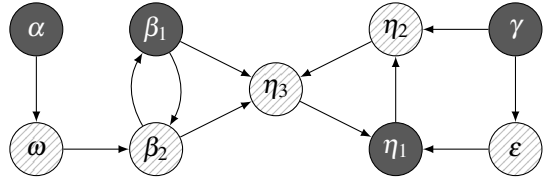
\begin{figure}[htbp]
        \centering
        \begin{tikzpicture}[scale=0.88,->,>=latex,
        auto,
        every node/.style={draw=black,circle,minimum size=0.7cm,inner sep=0pt},
    outnode/.style={text=black, pattern=north east lines, pattern color=gray!50},
    innode/.style={fill=gray!130, text=white},
    undnode/.style={fill=gray!50, text=white}
        ]
            \node[innode] (a1) at (-3.6,0.9) {$\alpha$};
            \node[outnode] (a2) at (-3.6,-0.9) {$\omega$};
            \node[innode] (b1) at (-1.8,0.9) {$\beta_{1}$};
            \node[outnode] (b2) at (-1.8,-0.9) {$\beta_{2}$};
            \node[outnode] (e3) at (0,0) {$\eta_{3}$};
            \node[innode] (e1) at (1.8,-0.9) {$\eta_{1}$};
            \node[outnode] (e2) at (1.8,0.9) {$\eta_{2}$};
            \node[outnode] (c2) at (3.6,-0.9) {$\epsilon$};
            \node[innode] (c1) at (3.6,0.9) {$\gamma$};
    
            \draw (a1) -- (a2);
            \draw (a2) -- (b2);
            \draw[bend left] (b2) to (b1);
            \draw[bend left] (b1) to (b2);
            \draw (b1) -- (e3);
            \draw (b2) -- (e3);
            \draw (e3) -- (e1);
            \draw (e1) -- (e2);
            \draw (e2) -- (e3);
            \draw (c1) -- (c2);
            \draw (c1) -- (e2);
            \draw (c2) -- (e1);
                     
        \end{tikzpicture}
        \caption{For Ex.~\ref{ex:forAC2m} ($\mathscr{F}_{3}$)}
        \label{fig:forAC2m}
\end{figure}

\begin{example}\label{ex:forAC2m}
    Consider the AF $\mathscr{F}_3$ and its $co$-labelling $\mathscr{L}$ in Fig.~\ref{fig:forAC2m}. Although $\alpha$, $\beta_1$, and $\gamma$ each have an even-length path to $\eta_1$, not all are actual causes: $\alpha$'s defense is ''blocked'' by $\beta_1$, whereas $\beta_1$ and $\gamma$ are both necessary for $\eta_1$'s acceptance.

    First, test $\mathbf{in}(\gamma)$. Intervening with $\mathbf{out}(\gamma)$ yields a intervened model $\mathscr{M}_{\mathscr{F}_3}^{\mathbf{out}(\gamma)}$, one of whose solution is $\mathscr{L}^{\prime}=\langle \{\epsilon,\eta_2,\beta_1,\alpha\},\{\gamma,\eta_1,\eta_3,\beta_2,\omega\},\emptyset \rangle$, where $\eta_1$'s label becomes $\mathbf{out}$. Hence, $\mathbf{in}(\gamma)$ is a but-for cause (and thus also satisfies $\mathrm{AC2}^m$).

    Next, test $\mathbf{in}(\beta_1)$. Under a but-for intervention $\mathbf{out}(\beta_1)$ alone, the unique labelling is $\mathscr{L}^{\prime\prime}=\langle\{\alpha,\beta_2,\eta_1,\gamma\},\{\omega,\beta_1,\eta_3,\eta_2,\epsilon\},\emptyset \rangle$,
    and $\eta_1$'s label remains. The blocked path from $\alpha$ is now active, preserving $\eta_1$'s acceptance.

    To prevent this, $\mathrm{AC2}^m$ also fixes $\beta_2$'s label to its original label $\mathbf{out}$. In $\mathscr{M}_{\mathscr{F}_3}^{\mathbf{out}(\beta_1)\wedge \mathbf{out}(\beta_2)}$, one solution is $\mathscr{L}^{\prime\prime\prime}=\langle \{\alpha,\gamma\},\{\omega,\beta_1,\beta_2,\epsilon\},\{\eta_1,\eta_2,\eta_3\} \rangle$,
    where $\eta_1$ is now labelled $\mathbf{und}$. Thus, $\mathbf{in}(\beta_1)$ qualifies as an actual cause under $\mathrm{AC2}^m$, while the stricter but-for condition $\mathrm{BFT}$ fails.
\end{example}

\subsection{Counterfactual Explanation}
\label{subsec:counterfactualexpl}

We now instantiate the formula-based explanation strategy $\operatorname{F-Expl}$ with \emph{counterfactual explanation} ($\operatorname{Counter}$) and \emph{but-for explanation} ($\operatorname{ButFor}$). The latter serves as a baseline to further justify the former as the appropriate instantiation.

While the notion of actual cause (Def.~\ref{def:actualcauses}) is independent of specific labels and accommodates any status of the topic argument ($\mathbf{in}$, $\mathbf{out}$, or $\mathbf{und}$), we restrict our focus here to explaining its \emph{acceptance} ($\mathbf{in}$). This restriction is motivated by the fact that standard formal properties (e.g., Relevance, Existence) and the prevailing explanation methods compared in Section~\ref{sec:relatedwork} are predominantly formulated for explaining \emph{accepted} arguments.

We adopt $\mathrm{AC2}^m$ for actual causes, and call the resulting method \emph{counterfactual explanation}. Let $\mathfrak{Act}(\mathscr{M}_{\mathscr{F}},\mathscr{L},\mathbf{in}(\alpha))$ be the set of hypotheses $\psi \in L_{\mathscr{UA}}^{-}$ that satisfy $\mathrm{AC2}^m_{\exists}$ (the strong version $\mathrm{AC2}^m_{\forall}$ is left for future work) for $\mathbf{in}(\alpha)$ under the solution $\mathscr{L}$ of $\mathscr{M}_{\mathscr{F}}$.

\begin{definition}[Counterfactual Explanation]
    Let $\mathscr{F}=(\mathscr{A},\mathscr{R})$ be an AF and $\mathscr{M}_{\mathscr{F}}$ its acceptability model. Consider $\alpha\in \mathscr{A}$ the topic, $\mathscr{L}$ a solution of $\mathscr{M}_{\mathscr{F}}$ with $\mathscr{L}(\alpha)=\mathbf{in}$. Then, 
    \[\operatorname{Counter}(\mathscr{F},\alpha)=\bigcup_{\mathscr{L}\in\mathsf{Sol}(\mathscr{M}_{\mathscr{F}})} \mathfrak{Act}(\mathscr{M}_{\mathscr{F}},\mathscr{L},\mathbf{in}(\alpha)).\]
\end{definition}


A basic requirement for any explanation is that the arguments it involves be relevant (recall Def.~\ref{def:path-relevant}) to the topic. The following proposition confirms that our counterfactual explanation method satisfies the \emph{Relevance} property.

\begin{proposition}\label{pro:relevant}
    Given an AF $\mathscr{F}=(\mathscr{A},\mathscr{R})$ and the topic $\alpha\in \mathscr{A}$. If the hypothesis $\psi\in \operatorname{Counter}(\mathscr{F},\alpha)$, then for all $\alpha_i\in \mathscr{A}_{\psi}$, it holds that $\alpha_i$ is relevant to $\alpha$.
\end{proposition}

\emph{Monotonicity} is a well-known property in argument-based explanation strategy~\cite{ulbrichtStrongExplanations2021}: if a set of arguments qualifies as an explanation, any superset should also be one. The counterfactual explanation method does not satisfy \emph{Monotonicity}, as Minimality requires that no proper sub-conjunction of the cause satisfies Actuality ($\mathrm{AC1}$) and Modified Counterfactual Condition ($\mathrm{AC2}^m$).

\begin{proposition}\label{pro:notMonotonicity}
    Given an AF $\mathscr{F}=(\mathscr{A},\mathscr{R})$ and the topic $\alpha\in \mathscr{A}$, for any hypotheses $\psi,\psi^{\prime}\in L_{\mathscr{UA}}^{-}$ with $\mathscr{A}_{\psi}\subsetneq \mathscr{A}_{\psi^{\prime}}$, if $\psi\in \operatorname{Counter}(\mathscr{F},\alpha)$ then $\psi^{\prime}\notin \operatorname{Counter}(\mathscr{F},\alpha)$.
\end{proposition}

For comparison, we also define the \emph{but-for explanation}.
Let $\mathfrak{BF}(\mathscr{M}_{\mathscr{F}},\mathscr{L},\mathbf{in}(\alpha))$ denote the set of all atomic hypotheses $\psi \in L_{\mathscr{UA}}^{-}$ that satisfy $\mathrm{BFT}_{\exists}$ for $\mathbf{in}(\alpha)$ under the solution $\mathscr{L}$ of $\mathscr{M}_{\mathscr{F}}$.
Consistent with the counterfactual explanation method, we restrict our focus to the weak version $\mathrm{BFT}_{\exists}$.
Then, $\operatorname{ButFor}(\mathscr{F},\alpha)=\bigcup _{\mathscr{L}\in\mathsf{Sol}(\mathscr{M}_{\mathscr{F}})} \mathfrak{BF}(\mathscr{M}_{\mathscr{F}},\mathscr{L},\mathbf{in}(\alpha))$.

The following proposition compares the two explanation methods w.r.t \emph{Non-trivial Explanation}: the existence of an explanation that involves accepted arguments beyond the topic itself when the topic is a sink node. It shows that counterfactual explanation always provides such an explanation, whereas the but-for explanation may fail to do so.

\begin{proposition}\label{pro:non-trivialexpl.}
    Let $\mathscr{F}=(\mathscr{A},\mathscr{R})$ be an AF and $\alpha\in \mathscr{A}$ the topic s.t. $\alpha^{-}\neq \emptyset$, $\alpha^{+}=\emptyset$ and $\mathscr{L}(\alpha)=\mathbf{in}$ for some $\mathscr{L}\in \mathscr{L}_{co}(\mathscr{F})$.
    \begin{enumerate}
        \item For any $\mathscr{F}$ and $\alpha$ satisfying the above conditions, there exists $\psi\in \operatorname{Counter}(\mathscr{F},\alpha)$, s.t. (i) $\mathscr{A}_{\psi}\setminus\{\alpha\}\neq\emptyset$, and (ii) for all $\beta\in \mathscr{A}_\psi$, $\mathscr{L}(\beta)=\mathbf{in}$.
        \item There exists $\mathscr{F}$ and $\alpha$ satisfying the above conditions, s.t. for all $\psi=\mathbf{in}(\beta)$, if $\psi\in\operatorname{ButFor}(\mathscr{F},\alpha)$, then $\mathscr{A}_{\psi}\setminus\{\alpha\}=\emptyset$.
    \end{enumerate}
\end{proposition}

\begin{example}\label{ex:overdetermination}
    Consider the AF $\mathscr{F}_{4}$ in Fig.~\ref{fig:overdetermination} and its unique $co$-labelling $\mathscr{L}$ depicted therein. Let the topic be $\eta$ labelled $\mathbf{in}$. Although $\eta$ is attacked by $\gamma$, it is accepted because both $\alpha$ and $\beta$ defend it, each sufficient to protect $\eta$ from $\gamma$'s attack.
     
    \begin{center}
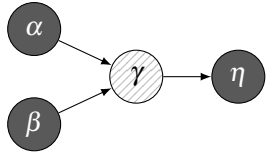
        \begin{minipage}[c]{0.4\textwidth}
            \centering
            \begin{tikzpicture}[->,>=latex,auto,
                every node/.style={draw=black,circle,minimum size=0.7cm,inner sep=0pt},
                outnode/.style={text=black, pattern=north east lines, pattern color=gray!50},
                innode/.style={fill=gray!130, text=white},
                undnode/.style={fill=gray!50, text=white},
                scale=0.9
            ]
                \node[outnode] (c) at (0,0) {$\gamma$};
                \node[innode] (a) at (-1.5,0.7) {$\alpha$};
                \node[innode] (b) at (-1.5,-0.7) {$\beta$};
                \node[innode] (d) at (1.5,0) {$\eta$};
                \draw (b)--(c);
                \draw (c)--(d);
                \draw (a)--(c);
                
            \end{tikzpicture}
            \captionof{figure}{$\mathscr{F}_4$: Overdetermination}
            \label{fig:overdetermination}
        \end{minipage}
        \hfill
        \begin{minipage}[c]{0.55\textwidth}
        \vspace{-10pt}
            Neither $\mathbf{in}(\alpha)$ nor $\mathbf{in}(\beta)$ passes the but-for test $\mathrm{BFT}_{\exists}$.
            By contrast, the compound hypothesis $\psi = \mathbf{in}(\alpha) \land \mathbf{in}(\beta)$ satisfies $\mathrm{AC2}^m_{\exists}$ and therefore belongs to $\operatorname{Counter}(\mathscr{F}_4,\eta)$, providing a non-trivial explanation with $\mathscr{A}_{\psi} \setminus \{\eta\} \neq \emptyset$.
        \end{minipage}
        \vspace{-3pt}
    \end{center}
      
\end{example}

\section{Graph Mutilations of Intervention}
\label{sec:IntandMod}

Following Pearl's graph mutilations for interventions in causal networks~\cite[p.~23]{pearl2009causality}, we define a corresponding graph operation on AFs for each atomic hypothesis $\mathbf{v}(\alpha)$.
These operations are illustrated in Fig.~\ref{fig:modification} and formalised in the following Def.~\ref{def:AtModi}.

Notably, this section provides a graphical counterpart of the intervention (Def.~\ref{def:Intervention}). The core contribution of the intervention and $\mathrm{AC2}^m$ lies in distinguishing causes from witnesses, a capability absent from earlier graph‑modification approaches and essential for moving beyond the but‑for test, as demonstrated in Section~\ref{sec:counterexpl}.

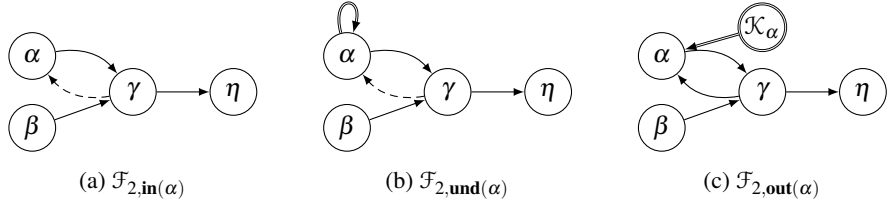
\begin{figure}[htbp]
    \centering
    \begin{subfigure}[b]{0.32\textwidth}
        \centering
        \begin{tikzpicture}[scale=0.95,->,>=latex,auto,transform shape,
            every node/.style={draw=black,circle,minimum size=0.65cm,inner sep=0pt},
            ]
            \node (c) at (0,0) {$\gamma$};
            \node (a) at (-1.4,0.5) {$\alpha$};
            \node (b) at (-1.4,-0.5) {$\beta$};
            \node (d) at (1.4,0) {$\eta$};
                 
            \draw (b)--(c);
            \draw (c)--(d);
            \draw (a) to[bend left] (c);
            \draw[densely dashed] (c) to[bend left] (a);
        \end{tikzpicture}
        \caption{$\mathscr{F}_{2,\mathbf{in}(\alpha)}$}
        \label{fig:mod-a}
    \end{subfigure}
    \hfill
    \begin{subfigure}[b]{0.32\textwidth}
        \centering
        \begin{tikzpicture}[scale=0.95,->,>=latex,auto,transform shape,
            every node/.style={draw=black,circle,minimum size=0.65cm,inner sep=0pt},
            ]
            \node (c) at (0,0) {$\gamma$};
            \node (a) at (-1.4,0.5) {$\alpha$};
            \node (b) at (-1.4,-0.5) {$\beta$};
            \node (d) at (1.4,0) {$\eta$};
                 
            \draw (b)--(c);
            \draw (c)--(d);
            \draw (a) to[bend left] (c);
            \draw[densely dashed] (c) to[bend left] (a);
            \draw[double, double distance=0.3pt] (a) to[loop above, min distance=0.55cm] (a);
        \end{tikzpicture}
        \caption{$\mathscr{F}_{2,\mathbf{und}(\alpha)}$}
        \label{fig:mod-b}
    \end{subfigure}
    \hfill
    \begin{subfigure}[b]{0.32\textwidth}
        \centering
        \begin{tikzpicture}[scale=0.95,->,>=latex,auto,transform shape,
            every node/.style={draw=black,circle,minimum size=0.65cm,inner sep=0pt},
            ]
            \node (c) at (0,0) {$\gamma$};
            \node (a) at (-1.4,0.5) {$\alpha$};
            \node (b) at (-1.4,-0.5) {$\beta$};
            \node (d) at (1.4,0) {$\eta$};
            \node[double, double distance=0.3pt] (e) at (0,0.9) {$\mathscr{K}_{\alpha}$};
                    
            \draw (b)--(c);
            \draw (c)--(d);
            \draw (a) to[bend left] (c);
            \draw (c) to[bend left] (a);
            \draw[double, double distance=0.3pt] (e)--(a);
        \end{tikzpicture}
        \caption{$\mathscr{F}_{2,\mathbf{out}(\alpha)}$}
        \label{fig:mod-c}
    \end{subfigure}
    \vspace{-8pt} 
    \caption{Atomic Mutilation of $\mathscr{F}_{2}$ in Fig.~\ref{fig:runningexample}. Symbols: dashed arrows $=$ removed attacks, double arrows $=$ newly introduced attacks, and double circles $=$ newly introduced arguments.}
    \label{fig:modification}
    \vspace{-3pt}
\end{figure}

\begin{definition}[Atomic Mutilation]\label{def:AtModi}
Given an AF $\mathscr{F}=(\mathscr{A},\mathscr{R})$ and an atomic hypothesis $\mathbf{v}(\alpha)$ of $\mathscr{F}$, three types of atomic mutilated $\mathscr{F}$ are defined as follows:
\begin{itemize}
    \item $\mathscr{F}_{\mathbf{out}(\alpha)} = \bigl( \mathscr{A} \cup \{ \mathscr{K}_\alpha \},\; \mathscr{R} \cup \{ (\mathscr{K}_\alpha,\alpha) \} \bigr)$,
    \item $\mathscr{F}_{\mathbf{in}(\alpha)}  = \bigl( \mathscr{A},\; \mathscr{R} \setminus \{ (\beta,\alpha) \mid \beta \in \alpha^{-} \} \bigr)$,
    \item $\mathscr{F}_{\mathbf{und}(\alpha)} = \bigl( \mathscr{A},\; \mathscr{R} \cup \{ (\alpha,\alpha) \} \setminus \{ (\beta,\alpha) \mid \beta \in \alpha^{-} \} \bigr)$.
\end{itemize}
\end{definition}

These atomic mutilations are not unique: $\mathscr{F}_{\mathbf{out}(\alpha)}$ and $\mathscr{F}_{\mathbf{in}(\alpha)}$ match Sakama~\cite{SakamaCounterfactualReasoning2014}, while $\mathscr{F}_{\mathbf{und}(\alpha)}$ follows Rienstra~\cite{rienstra2014phdthesisarinflux}. Unlike their metaphysical interpretations, our intervention is purely instrumental for probing dependencies, treating the graph representation merely as a technical convenience. Any admissible formulation must respect two properties:

\vspace{-5pt}
\begin{description}
    \item[Node Independence]:  In $\mathscr{M}^{\mathbf{v}(\alpha)}_{\mathscr{F}}$, the label of $\alpha$ no longer depends on its attackers. This is realized in the graph mutilation as follows: for $\mathscr{F}_{\mathbf{out}(\alpha)}$, adding $\mathscr{K}_\alpha$ overrides all incoming attacks; for $\mathscr{F}_{\mathbf{in}(\alpha)}$, incoming attacks are removed; for $\mathscr{F}_{\mathbf{und}(\alpha)}$, incoming attacks are removed and a self‑attack is added.
    \item[Status Uniqueness]: The intervention fixes the label of $\alpha$ to $\mathbf{v}$ in $\mathscr{M}^{\mathbf{v}(\alpha)}_{\mathscr{F}}$. The same holds in the mutilated AF.
\end{description}
\vspace{-5pt}

The two properties discussed above jointly guarantee that the application of intervention hypotheses is order-invariant.

\begin{proposition}\label{pro:order0invariant}
Let $\mathscr{M}_\mathscr{F}$ be an acceptability model and $\psi=\psi_1 \land \psi_2$ a compound hypothesis. Then, $\mathscr{M}^{\psi}_\mathscr{F} \;=\; \big(\mathscr{M}_{\mathscr{F}}^{\psi_1}\big)^{\psi_2} \;=\; \big(\mathscr{M}_{\mathscr{F}}^{\psi_2}\big)^{\psi_1}$.

\end{proposition}

This result also holds for framework mutilation. Based on this, we define mutilation of an AF by a compound hypothesis through the recursion of atomic mutilations.

\begin{definition}[Compound Mutilation]
Let $\mathscr{F}$ be an AF and $\psi\in L_{\mathscr{UA}}^{-}$ a hypothesis. The compound mutilated AF $\mathscr{F}_\psi$ is defined recursively w.r.t the possible forms of $\psi$:
\begin{itemize}    
\item If $\psi=\mathbf{v}(\alpha)$, then $\mathscr{F}_\psi = \mathscr{F}_{\mathbf{v}(\alpha)}$,
\item If $\psi=\psi_1\land \psi_2$, then $\mathscr{F}_\psi = (\mathscr{F}_{\psi_1})_{\psi_2}=(\mathscr{F}_{\psi_2})_{\psi_1}$.
\end{itemize}
\end{definition}

We end this section by establishing a mapping from solutions of the intervened models to labellings of the mutilated AFs. 

\begin{theorem}\label{the:intervention-modification}
    Let $\mathscr{F}=(\mathscr{A},\mathscr{R})$ be an AF, $\mathscr{M}_\mathscr{F}$ its acceptability model, and $\psi\in L_{\mathscr{UA}}^{-}$ a hypothesis. For every $\mathscr{L}\in \mathsf{Sol}(\mathscr{M}_{\mathscr{F}}^{\psi})$, there exists a labelling $\mathscr{L}'$ of the mutilated AF $\mathscr{F}_{\psi}$ s.t., for any argument $\alpha \in \mathscr{A}$, it holds that $\mathscr{L}(\alpha)=\mathscr{L}'(\alpha)$. 
\end{theorem}

\section{Discussion}
\label{sec:relatedwork}

Most explanation methods in the literature are \emph{argument-based}, defined as a function $\operatorname{A-Expl}: \mathscr{UF} \times \mathscr{UA} \to 2^{2^{\mathscr{A}}}$ that returns sets of arguments as explanations. We recall four prominent instantiations $\operatorname{Root}_{\sigma}$, $\operatorname{Suf}$, $\operatorname{Nec}$, and $\operatorname{Close-Count}_{\sigma}$ below.

\paragraph{\textbf{Root Reason.}}
Liao and Van Der Torre~\cite[Def.~14]{liaoAttackDefenseSemantics2024} introduce \emph{root reason}. 
Given an AF $\mathscr{F}=(\mathscr{A},\mathscr{R})$, an extension $E\in \mathscr{E}_{\sigma}(\mathscr{F})$, and a topic $\alpha\in E$, let $\overline{E}=\{(\gamma,\delta)\mid \gamma,\delta\in E,\gamma\text{ directly defends }\delta\}$ with transitive closure $\overline{E}^{+}$. 
An argument $\beta\in E$ is \emph{initial} if it has no attackers in $\mathscr{F}$, and \emph{self-explanatory} if $(\beta,\beta)\in\overline{E}^{+}$ (denoted $E_I$ and $E_S$). 
Then, 
\[\operatorname{Root}_{\sigma}(\mathscr{F},\alpha)=\bigl\{\, \{\beta\in E\mid (\beta,\alpha)\in \overline{E}^{+},\ \beta\in E_I\cup E_S\} \bigm| E\in \mathscr{E}_{\sigma}(\mathscr{F}),\ \alpha\in E \,\bigr\}.\]

\paragraph{\textbf{Sufficient and Necessary Explanation.}}
Borg and Bex~\cite[Def.~28]{borgMNS2024} formalise \emph{sufficient and necessary} explanations in AA using the notion of relevance (Def.~\ref{def:path-relevant}).
A set $S\subseteq \mathscr{A}$ is \emph{sufficient} for $\alpha$ iff $S$ is relevant for $\alpha$, conflict-free, and defends $S\cup\{\alpha\}$. 
An argument $\beta\in \mathscr{A}$ is \emph{necessary} for $\alpha$ iff $\beta$ is relevant for $\alpha$ and, for every admissible set $E\in \mathscr{E}_{ad}(\mathscr{F})$, $\beta\notin E$ implies $\alpha\notin E$. Then, 
\begin{align*}
\operatorname{Suf}(\mathscr{F},\alpha) &= \bigl\{\, S\cup\{\alpha\} \bigm| S\text{ is sufficient for }\alpha \,\bigr\},\\
\operatorname{Nec}(\mathscr{F},\alpha) &= \bigl\{\, \{\alpha\}\cup\{\beta\in\mathscr{A}\mid \beta\text{ is necessary for }\alpha\} \,\bigr\}.
\end{align*}


The next proposition relates counterfactual explanation to root reasons and sufficient explanation.
When causes are restricted to initial or self‑explanatory arguments, counterfactual explanation is more precise than root reasons.
When restricted to accepted arguments, it is more selective than the sufficient explanation.

\begin{proposition}\label{pro:counter-root}
    Let $\mathscr{F}=(\mathscr{A},\mathscr{R})$ be an AF, $\alpha\in \mathscr{A}$ the topic, and $\mathscr{L}\in \mathscr{L}_{co}(\mathscr{F})$ a $co$-labelling with $\mathscr{L}(\alpha)=\mathbf{in}$. Then,
    \begin{enumerate}
        \item For all $\psi\in \operatorname{Counter}(\mathscr{F},\alpha)$ with $\mathscr{A}_{\psi}\subseteq \mathbf{in}(\mathscr{L})_{I}\cup \mathbf{in}(\mathscr{L})_{S}$ and $(\alpha_i,\alpha)\in \overline{\mathbf{in}(\mathscr{L})}^{+}$ for all $\alpha_i\in \mathscr{A}_{\psi}$, there exists $E\in \operatorname{Root}_{co}(\mathscr{F},\alpha)$ s.t. $\mathscr{A}_{\psi}\subseteq E$. The converse does not hold.
        
        \item For all $\psi\in\operatorname{Counter}(\mathscr{F},\alpha)$ with $\mathscr{A}_{\psi}\subseteq\mathbf{in}(\mathscr{L})$, there exists $E\in\operatorname{Suf}(\mathscr{F},\alpha)$ s.t. $\mathscr{A}_{\psi}\subseteq E$. The converse does not hold: there exist $E\in\operatorname{Suf}(\mathscr{F},\alpha)$ for which no such $\psi$ satisfies $\mathscr{A}_{\psi}\subseteq E\setminus\{\alpha\}$.
    \end{enumerate}



\end{proposition}

\begin{example}[Cont.~of Ex.~\ref{ex:runningexample}]\label{ex:selective}
    Recall the AF $\mathscr{F}_{2}$ in Fig.~\ref{fig:runningexample} and its unique $co$-labelling $\mathscr{L}$ depicted therein, with $\eta$ the topic argument. $\operatorname{Root}_{co}(\mathscr{F}_{2},\eta)=\{\{\alpha,\beta\}\}$, since $\alpha\in \mathbf{in}(\mathscr{L})_{S}$ and $\beta\in \mathbf{in}(\mathscr{L})_{I}$, and both directly defend $\eta$. $\operatorname{Suf}(\mathscr{F}_2,\eta)=\{\{\alpha,\beta,\eta\},\{\alpha,\eta\},\{\beta,\eta\}\}$, each $S\setminus \{\eta\}$ for $S\in \operatorname{Suf}(\mathscr{F}_2,\eta)$ being conflict-free, relevant, and defending $S\cup \{\eta\}$. 

    For counterfactual explanation, $\mathbf{in}(\beta)$ is an actual cause, whereas $\mathbf{in}(\alpha)$ fails $\mathrm{AC2}^m_{\exists}$.
    First, after intervening with $\mathbf{out}(\alpha)$, $\mathscr{M}_{\mathscr{F}_2}^{\mathbf{out}(\alpha)}$ has the unique solution $\mathscr{L}_1=\langle \{\beta,\eta\},\{\alpha,\gamma\},\emptyset \rangle$, where $\eta$ remains $\mathbf{in}$.
    Second, after intervening with $\mathbf{und}(\alpha)$, $\mathscr{M}_{\mathscr{F}_2}^{\mathbf{und}(\alpha)}$ has the unique solution $\mathscr{L}_2=\langle \{\beta,\eta\},\{\gamma\},\{\alpha\}\rangle$, with $\eta$ still $\mathbf{in}$.
    Third, no admissible witness can rescue $\mathbf{in}(\alpha)$: in both solutions, $\gamma$ retains its original label $\mathbf{out}$, so fixing any other arguments as witnesses does not alter the outcome.
    In contrast, $\mathscr{M}_{\mathscr{F}_2}^{\mathbf{out}(\beta)}$ admits a solution $\mathscr{L}_3=\langle \emptyset, \{\beta\}, \{\alpha,\gamma,\eta\} \rangle$, where $\eta$ is $\mathbf{und}$. Hence $\mathbf{in}(\beta)$ qualifies as an actual cause.
    Furthermore, Minimality excludes any $\psi^{\prime}$ with $\mathscr{A}_{\psi^{\prime}}\supsetneq\mathscr{A}_{\mathbf{in}(\beta)}$; hence, $\psi^{\prime}\notin\operatorname{Counter}(\mathscr{F}_2,\eta)$.

    Thus, for $\psi=\mathbf{in}(\beta)\in \operatorname{Counter}(\mathscr{F}_2,\eta)$, $\mathscr{A}_{\psi}=\beta$ is a subset of every $E\in \operatorname{Root}_{co}(\mathscr{F}_{2},\eta)$ and every $S\in \operatorname{Suf}(\mathscr{F}_2,\eta)$, yet the converse fails.
\end{example}

\begin{figure}[htbp]
        \centering
        \begin{tikzpicture}[scale=0.88,->,>=latex,
        auto,
        every node/.style={draw=black,circle,minimum size=0.7cm,inner sep=0pt},
    outnode/.style={text=black, pattern=north east lines, pattern color=gray!50},
    innode/.style={fill=gray!130, text=white},
    undnode/.style={fill=gray!50, text=white}
        ]
            \node[innode] (a1) at (-3.8,0) {$\alpha_{1}$};
            \node[outnode] (a2) at (-2.8,0.9) {$\alpha_{2}$};
            \node[innode] (a3) at (-1.8,0) {$\alpha_{3}$};
            \node[outnode] (a4) at (-2.8,-0.9) {$\alpha_{4}$};

            \draw (a1) -- (a2);
            \draw (a2) -- (a3);
            \draw (a3) -- (a4);
            \draw (a4) -- (a1);

            \node[outnode] (b) at (0,0) {$\beta$};

            \node[innode] (y1) at (1.5,0.9) {$\gamma_1$};
            \node[outnode] (y2) at (3,0.9) {$\gamma_2$};

            \node[innode] (c) at (2.25,-0.5) {$\eta$};
            \node[innode] (e) at (3.75,-0.5) {$\epsilon$};

            \draw (a3) -- (b);
            \draw (b) -- (y1);
            \draw (y1) to[bend left] (y2);
            \draw (y2) to[bend left] (y1);
            \draw (c) -- (y2);
            \draw (e) -- (y2);

        \end{tikzpicture}
        \caption{For Ex.~\ref{ex:rewithSuffandRoot} ($\mathscr{F}_{5}$)}
        \label{fig:rewithSuffandRoot}
\end{figure}
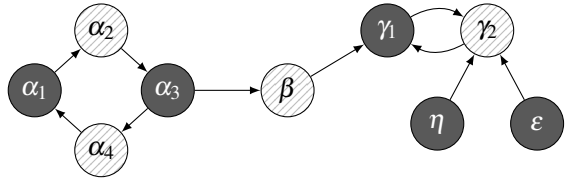

\begin{example}\label{ex:rewithSuffandRoot}
    Consider the AF $\mathscr{F}_5$ and the $co$-labelling $\mathscr{L}$ in Fig.~\ref{fig:rewithSuffandRoot}, with $\gamma_1$ as the topic argument.
    The set $E=\{\alpha_1,\alpha_3,\gamma_1\}$ is a sufficient explanation. And the set of $\mathbf{in}$ labelled arguments $E^{\prime}=\mathbf{in}(\mathscr{L})$ is a root reason. $\psi_1=\mathbf{in}(\alpha_1)$, $\psi_2=\mathbf{in}(\alpha_3)$, and $\psi_3=\mathbf{in}(\eta)\wedge \mathbf{in}(\epsilon)$ are all actual causes.
    Moreover, $\mathscr{A}_{\psi_i}$ is the subset of both $E$ and $E^{\prime}$ for $i\in \{1,2,3\}$.
\end{example}

\paragraph{\textbf{Closest-World-Based Counterfactual Explanation.}} In contrast to our intervention-based framework, Alfano et al.~\cite[Def. 1]{alfanoCounterfactual2024} introduce an alternative counterfactual explanation rooted in the \emph{closest-world semantics}~\cite{lewis1973counterfactuals}.
In this approach, a labelling $\mathscr{L}$ of an AF $\mathscr{F}=(\mathscr{A},\mathscr{R})$ is seen as the \emph{actual world} and other $\sigma$-labellings as \emph{possible worlds}. The similarity between two worlds is measured by the Hamming distance $\delta(\mathscr{L},\mathscr{L}') = |\{\alpha\in \mathscr{A} \mid \mathscr{L}(\alpha) \neq \mathscr{L}'(\alpha)\}|$.
A labelling $\mathscr{L}^{\prime}$ is a counterfactual of $\mathscr{L}$ w.r.t. $\alpha$ iff: 
\begin{enumerate}
    \item $\mathscr{L}(\alpha)\neq \mathscr{L}^{\prime}(\alpha)$, and
    \item there exists no $\mathscr{L}^{\prime\prime}\in \mathscr{E}_{\sigma}(\mathscr{F})$ s.t. $\mathscr{L}(\alpha)\neq \mathscr{L}^{\prime\prime}(\alpha)$ and $\delta(\mathscr{L},\mathscr{L}^{\prime\prime})<\delta(\mathscr{L},\mathscr{L}^{\prime})$.
\end{enumerate}

The set of counterfactuals of $\mathscr{L}$ in an AF w.r.t. $\alpha$ is denoted as $\mathfrak{CF}^{\sigma}(\alpha,\mathscr{L})$, and each such counterfactual labelling serves as an explanation for the topic's acceptance. Then,
\[
\operatorname{Close-Count}_{\sigma}(\mathscr{F},\alpha) =\big\{\mathbf{in}(\mathscr{L}^{\prime})\mid
    \mathscr{L}\in \mathscr{E}_{\sigma}(\mathscr{F}), \alpha\in \mathbf{in}(\mathscr{L}),
    \text{and}~\exists \mathscr{L}^{\prime}~\text{s.t.}~\mathscr{L}^{\prime}\in \mathfrak{CF}^{\sigma}(\alpha,\mathscr{L})\big\}.
\]

Following the methodology of actual causality~\cite{beckersprincipled2018}, human intuitions about causes can be formalised as \emph{structural principles} on a given graph topology.
While an exhaustive principle-based analysis is beyond the scope of this discussion,
Fig.~\ref{fig:preemption} (Preemption) and Fig.~\ref{fig:overdetermination} (Overdetermination) instantiate the key structural patterns behind those principles.
The next example compares our method with $\operatorname{Root}_{co}$, $\operatorname{Suf}$, $\operatorname{Nec}$, and $\operatorname{Close-Count}_{co}$ on these two structures.

\begin{example}[$\operatorname{Root_{co}}$, $\operatorname{Suf}$, $\operatorname{Nec}$, $\operatorname{Close-Count}_{co}$ vs. $\operatorname{Counter}$]\label{ex:root-suf-nec}
    Consider again the AFs $\mathscr{F}_1$ (Preemption, Fig.~\ref{fig:preemption}) and $\mathscr{F}_{4}$ (Overdetermination, Fig.~\ref{fig:overdetermination}), each with the unique $co$-labelling depicted therein. Let the topic be $\eta$ in both AFs.
    \vspace{-5pt}
    \begin{description}
    \item[Preemption ($\mathscr{F}_1$):]
        $\operatorname{Root}_{co}(\mathscr{F}_1,\eta)=\{\{\alpha\}\}$,
        $\{\alpha,\beta,\eta\}\in\operatorname{Suf}(\mathscr{F}_1,\eta)$,
        $\operatorname{Nec}(\mathscr{F}_1,\eta)=\\\{\{\alpha,\eta\}\}$, $\operatorname{Close-Count}_{co}(\mathscr{F}_1,\eta)=\emptyset$,
        $\mathbf{in}(\alpha)\in \operatorname{Counter}(\mathscr{F}_1,\eta)$ while $\mathbf{in}(\beta)\notin \operatorname{Counter}(\mathscr{F}_1,\eta)$.
    
    \item[Overdetermination ($\mathscr{F}_4$):]
        $\operatorname{Root}_{co}(\mathscr{F}_4,\eta)=\{\{\alpha,\beta\}\}$,
        $\{\alpha,\eta\},\{\beta,\eta\}\in\operatorname{Suf}(\mathscr{F}_4,\eta)$,
        $\operatorname{Nec}(\mathscr{F}_4,\eta)=\{\{\eta\}\}$, $\operatorname{Close-Count}_{co}(\mathscr{F}_{4},\eta)=\emptyset$,
        $\mathbf{in}(\alpha)\wedge\mathbf{in}(\beta)\in\operatorname{Counter}(\mathscr{F}_4,\eta)$ while $\mathbf{in}(\alpha),\mathbf{in}(\beta)\notin\operatorname{Counter}(\mathscr{F}_4,\eta)$.
    \end{description}
    \vspace{-5pt}
Three observations follow.
\begin{enumerate}
    \item $\operatorname{Root}_{co}$ performs well in Preemption (singling out $\alpha$) but fails in Overdetermination, as it cannot distinguish between the two symmetric defenders $\alpha$ and $\beta$.
    \item $\operatorname{Nec}$ is closer to $\operatorname{Counter}$ than $\operatorname{Suf}$ is: in $\mathscr{F}_{1}$, $\operatorname{Nec}$ returns $\{\alpha,\eta\}$, whose non-topic element matches the actual cause $\mathbf{in}(\alpha)$; in $\mathscr{F}_4$, $\operatorname{Nec}$ returns $\{\eta\}$, reflecting the absence of any single necessary argument, while $\operatorname{Counter}$ captures joint necessity via $\mathbf{in}(\alpha)\wedge \mathbf{in}(\beta)$. This confirms that $\operatorname{Counter}$ merely tracks necessity, not sufficiency, which is a perspective largely overlooked in prior work.
    \item $\operatorname{Close-Count}_{co}$ fails to satisfy \emph{Existence}, which is a foundational principle for any explanation method. Because both $\mathscr{F}_1$ and $\mathscr{F}_4$ admit only a unique $co$-labelling, there are no alternative ``possible worlds (counterfactual labellings)'' where the status of $\eta$ changes. Consequently, $\operatorname{Close-Count}_{co}$ yields an empty set ($\emptyset$) for both frameworks, which directly contradicts the intuitive premise that the acceptance of $\eta$ has identifiable causes.
\end{enumerate}

\end{example}

\section{Conclusion and Future Work}
\label{sec:con}
We introduce an \emph{intervention-based} counterfactual reasoning framework in abstract argumentation and develop counterfactual explanations within this setting. By fixing \emph{witness} arguments, our method overcomes the limitations of the but-for approach. We further show that the intervention operation can be visualized through graph mutilation, thereby improving our method's explainability for human-aligned interaction.
A comparison with existing argument-based explanation methods shows that our counterfactual explanation is more selective while maintaining connections to these approaches.

Several research directions follow naturally. 
\emph{First}, our method and existing approaches warrant a thorough principle-based analysis. This includes different paradigms of counterfactual reasoning, such as the \emph{non-intervention-based} counterfactual explanation of Alfano et al.~\cite{alfanoCounterfactual2024} rooted in \emph{closest-world} semantics, and the \emph{post-hoc} explanation proposed by Amgoud~\cite{amgoud2024posthoc}.
We plan to develop \emph{structural principles} that describe specific graph patterns, thereby examining how these methods perform under those patterns.
\emph{Second}, the condition $\mathrm{AC2}^m$ is rooted in \emph{acyclic} structural equation models. Although it fits well, we will explore counterfactual conditions that are more compatible with \emph{cyclic} argumentation frameworks. The intervention-based counterfactual reasoning framework introduced in this paper offers an expressive logical language for this endeavor.
\emph{Third}, an analysis of the computational complexity and the development of algorithmic implementations for our framework remain crucial directions for practical deployment.
\emph{Fourth}, we plan to extend the approach to other semantics, and to extended frameworks (PAF, ADF)~\cite{alfanoExplainableAcceptanceProbabilistic2020,TjitzeADF2022}, as well as structured frameworks (ASPIC\textsuperscript{+}, ABA)~\cite{BengelCausal2022ArgXAI,Pisanocauseinfact2025,bochmancausationargumentation2025}.


\section*{Acknowledgements}
The authors are thankful to the anonymous reviewers for their helpful comments and suggestions. This work is supported by the National Natural Science Foundation of China (No. 62576309).

\bibliographystyle{vancouver}
\bibliography{refs}

\clearpage

\appendix
\begin{center}
    \large \textbf{Appendix: Proofs}
\end{center}

\section*{Proof of Theorem~\ref{the:SolandL}.}
\begin{proof}
\textbf{($\Rightarrow$) $\mathsf{Sol}(\mathscr{M}_{\mathscr{F}})\subseteq \mathscr{L}_{co}(\mathscr{F})$:}
Let $\mathscr{L}\in \mathsf{Sol}(\mathscr{M}_{\mathscr{F}})$ be a solution of the acceptability model $\mathscr{M}_{\mathscr{F}}$. Take an arbitrary argument $\alpha\in \mathscr{A}$. Since $\mathscr{L}$ satisfies every equation in $\mathscr{M}_{\mathscr{F}}$, we have $\mathscr{L}(\alpha)=f_{\alpha}\big((\mathscr{L}(\beta)_{\beta\in \alpha^{-}})\big)$, where $f_{\alpha}$ is the acceptability equation of $\alpha$.

According to Def.~\ref{def:acceptabilityModel}, the output of the function $f_{\alpha}$ is completely determined by the labels of the attackers $\beta\in \alpha^{-}$. Now, examine the three possible cases for $\mathscr{L}(\alpha)$:
\begin{description}
    \item[Case 1] If $\mathscr{L}(\alpha)=\mathbf{in}$, then by the equation, every attacker $\beta\in \alpha^{-}$ must satisfy $\mathscr{L}(\beta)=\mathbf{out}$. Hence, condition (ii) of a complete labelling holds for $\alpha$.
    \item[Case 2] If $\mathscr{L}(\alpha)=\mathbf{out}$, then there exists $\beta\in \alpha^{-}$ with $\mathscr{L}(\beta)=\mathbf{in}$. Thus, condition (i) of a complete labelling holds for $\alpha$.
    \item[Case 3] If $\mathscr{L}(\alpha)=\mathbf{und}$, then it is neither the case that all are $\mathbf{out}$ nor that some attacker is $\mathbf{in}$, which indicates that $\alpha$ must satisfy the condition (i) and (ii) of a complete labelling.
\end{description}

Since $\alpha$ was arbitrary, $\mathscr{L}$ satisfies the complete labelling conditions for every argument, i.e., $\mathscr{L}\in \mathscr{L}_{co}(\mathscr{F})$. Given that $\mathscr{L}$ is an arbitrary element of $\mathsf{Sol}(\mathscr{M}_\mathscr{F})$, we conclude that $\mathsf{Sol}(\mathscr{M}_\mathscr{F}) \subseteq \mathscr{L}_{co}(\mathscr{F})$.
\vspace{1.5em}

\textbf{($\Leftarrow$) $\mathscr{L}_{co}(\mathscr{F}) \subseteq \mathsf{Sol}(\mathscr{M}_{\mathscr{F}})$:}
Let $\mathscr{L}\in \mathscr{L}_{co}(\mathscr{F})$ be a complete labelling of $\mathscr{F}$. For any $\alpha\in \mathscr{A}$, we show that $\mathscr{L}(\alpha)=f_{\alpha}\big((\mathscr{L}(\beta)_{\beta\in \alpha^{-}})\big)$:
\begin{description}
    \item[Case 1] If $\mathscr{L}(\alpha)=\mathbf{in}$, then every attacker $\beta\in \alpha^{-}$ is labelled $\mathbf{out}$. According to Def.~\ref{def:acceptabilityModel}, $f_{\alpha}$ returns $\mathbf{in}$ in this situation.
    \item[Case 2] If $\mathscr{L}(\alpha)=\mathbf{out}$, then there exists an attacker $\beta\in \alpha^{-}$ with $\mathscr{L}(\beta)=\mathbf{in}$. By Def.~\ref{def:acceptabilityModel}, $f_{\alpha}$ returns $\mathbf{out}$.
    \item[Case 3] If $\mathscr{L}(\alpha)=\mathbf{und}$, then it is not true that all attackers are $\mathbf{out}$ and it is also not true that some attackers is $\mathrm{in}$. Hence, by Def.~\ref{def:acceptabilityModel}, $f_{\alpha}$ returns $\mathbf{und}$.
\end{description}

Thus, $\mathscr{L}$ satisfies the acceptability equation of $\alpha$ for all $\alpha\in \mathscr{A}$. Consequently, $\mathscr{L}$ is a solution of $\mathscr{M}_{\mathscr{F}}$, i.e., $\mathscr{L}\in \mathsf{Sol}(\mathscr{M}_{\mathscr{F}})$. Given that $\mathscr{L}$ is an arbitrary element of $\mathscr{L}_{co}(\mathscr{F})$, we conclude that $\mathscr{L}_{co}(\mathscr{F}) \subseteq \mathsf{Sol}(\mathscr{M}_{\mathscr{F}})$

\vspace{0.5em}
Therefore, $\mathsf{Sol}(\mathscr{M}_{\mathscr{F}})=\mathscr{L}_{co}(\mathscr{F})$ for every AF $\mathscr{F}$.
    
\end{proof}

\section*{Proof of Theorem~\ref{theorem:StrongtoWeak}}

\begin{proof}
     Suppose that $\psi$ is a but-for cause for $\mathbf{v}(\alpha)$ under $\mathrm{BFT}_{\forall}$.
     It is a fundamental result in abstract argumentation that every argumentation framework admits at least one complete extension (specifically, the grounded extension is always a complete extension). Consequently, the set of complete extensions is strictly non-empty. In any non-empty domain, universal quantification logically entails existential quantification.
     Therefore, by Def.~\ref{def:Model-Satisfaction}, $\mathscr{M}\models_{\forall}\phi$ implies $\mathscr{M}\models_{\exists}\phi$ for any arbitrary formula $\phi$. Hence, the counterfactual condition of $\mathrm{BFT}_{\forall}$ directly entails the corresponding weak version $\mathrm{BFT}_{\exists}$ for the same modified hypothesis $\psi^{\prime}$.
     By an identical line of reasoning, this implication strictly holds from $\mathrm{AC2}^m_{\forall}$ to $\mathrm{AC2}^m_{\exists}$.
\end{proof}

\section*{Proof of Theorem~\ref{theorem:threevariants}}

\begin{proof}
    Assume that $\psi$ is an atomic hypothesis and that $\psi$ is a but-for cause for $\mathbf{v}(\alpha)$ under $\mathrm{BFT}_q$ for some $q\in \{\forall,\exists\}$. By Def.~\ref{def:butforcause}, 
    \[
    \exists \psi^{\prime}\in \mathsf{Mod}_{\mathbf{in/out}}(\psi): \mathscr{M}_{\mathscr{F}}\models_{q}[\psi^{\prime}]\neg \mathbf{v}(\alpha) \text{ for } q\in \{\forall,\exists\}.
    \]

    We now verify that $\psi$ satisfies $\mathrm{AC2}^m_{q}$.

    Choose $\chi$ to be the empty conjunction, i.e., a tautology $\top$. Trivially, $(\mathscr{M}_{\mathscr{F}},\mathscr{L})\Vdash\chi$ holds for any solution $\mathscr{L}$. Since $\psi^{\prime}\wedge \chi$ is logically equivalent to $\psi^{\prime}$, the intervened model $\mathscr{M}_{\mathscr{F}}^{\psi^{\prime}\wedge \chi}$ is identical to $\mathscr{M}_{\mathscr{F}}^{\psi^{\prime}}$. Therefore,
    \[
     \exists \psi^{\prime}\in \mathsf{Mod}_{\mathbf{in/out}}(\psi):\mathscr{M}_{\mathscr{F}}\models_{q}[\psi^{\prime}\wedge \chi]\neg \mathbf{v}(\alpha)
    \]
    holds as well. Note that $\mathsf{Mod}_{\mathbf{in/out}}(\psi)$ is a subset of $\mathsf{Mod}(\psi)$, it holds that
  \[
     \exists \psi^{\prime}\in \mathsf{Mod}(\psi):\mathscr{M}_{\mathscr{F}}\models_{q}[\psi^{\prime}\wedge \chi]\neg \mathbf{v}(\alpha).
    \]

    Hence, $\psi$ meets condition $\mathrm{AC2}^m_{q}$ with the same $\psi^{\prime}$ and the chosen $\chi=\top$. Conditions $\mathrm{AC1}$ and $\mathrm{AC3}$ are already satisfied because (1) $\mathrm{AC1}$ is given in the condition of $\mathrm{BFT}_q$ where $(\mathscr{M},\mathscr{L})\Vdash \psi\wedge \mathbf{v}(\alpha)$, and (2) $\psi$ is an atomic hypothesis, which suggests that $\psi$ is minimal w.r.t set inclusion. Consequently, $\psi$ is an actual cause for $\mathbf{v}(\alpha)$ under $\mathrm{AC2}^m_{q}$.

\end{proof}

\section*{Proof of Theorem~\ref{the:intervention-modification}}

\begin{proof}

Let $\mathscr{L}\in \mathsf{Sol}(\mathscr{M}^\psi_{\mathscr{F}})$ be a solution of the intervened acceptability model $\mathscr{M}^{\psi}_{\mathscr{F}}$. Take an arbitrary argument $\alpha\in \mathscr{A}$. Since $\mathscr{L}$ satisfies every equation in $\mathscr{M}^{\psi}_{\mathscr{F}}$, we have 
\[
\mathscr{L}(\alpha)=f^{\psi}_{\alpha}\big((\mathscr{L}(\beta)_{\beta\in \alpha^{-}})\big)
\]
where $f^{\psi}_{\alpha}$ is the intervened acceptability equation of $\alpha$.

We prove that $\mathscr{L}$ satisfies the complete labelling conditions for every argument $\alpha \in \mathscr{A}$. With additional arguments such as $\mathscr{K}_\alpha$ always labeled $\mathbf{in}$, a complete labelling $\mathscr{L}'$ of the modified AF $\mathscr{F}_{\psi}$ could be construct from $\mathscr{L}$.

We prove this by cases.

\begin{description}
    \item[Case 1] For intervened arguments, which suggests there exists $\mathbf{v}\in \mathbb{V}$ such that $\mathbf{v}(\alpha) \in \mathsf{Sub}(\psi)$, therefore the equation $f^{\psi}_{\alpha}=\mathbf{v}$. Now, we examine the three possible cases in the modified AF $\mathscr{F}_\psi$ to show that in every complete labelling of $\mathscr{F}_\psi$, $\alpha$ is always labeled $\mathbf{v}$:
    \begin{description}
    \item[1.1] If $\mathbf{v}=\mathbf{out}$, then there is an additional initial argument $\mathscr{K}_\alpha$ attacks $\alpha$. Given that $\mathscr{K}_\alpha \not \in  \mathscr{A}$, $\mathscr{K}_\alpha$ could not be intervened according to Def.~\ref{def:Intervention}. So $\mathscr{K}_\alpha$ is initial in $\mathscr{F}_\psi$, and $\alpha$ is attacked by an initial argument $\mathscr{K}_\alpha$, so $\alpha$ is labeled $\mathbf{out}$ in every complete labelling of $\mathscr{F}_\psi$.
    \item[1.2] If $\mathbf{v}=\mathbf{in}$, then all attacks from its parent nodes are removed from $\mathscr{F}_\psi$, which suggests $\alpha$ is now an initial argument in $\mathscr{F}_\psi$, this leads to the consequence that $\alpha$ is labeled $\mathbf{in}$ in every complete labelling of $\mathscr{F}_\psi$.
    \item[1.3] If $\mathbf{v}=\mathbf{und}$, then all attacks from its parent nodes are removed from $\mathscr{F}_\psi$, and a new self-attack is introduced. In every complete labelling of $\mathscr{F}_\psi$, $\alpha$ could not be labeled $\mathbf{in}$, since one of its attacker, $\alpha$, is labeled $\mathbf{in}$; $\alpha$ could not be labeled $\mathbf{out}$ either, because in that case all attackers of $\alpha$ are labeled $\mathbf{out}$, which suggests $\alpha$ should be labeled $\mathbf{in}$. So $\alpha$ could only be labeled $\mathbf{und}$ in every complete labelling of $\mathscr{F}_\psi$.
\end{description}
\item[Case 2] For non-intervened arguments, the intervened function $f^{\psi}_{\alpha}=f_{\alpha}$, which suggests that the label of $\alpha$ is completely determined by the labels of the attackers $\beta\in \alpha^{-}$. Now, examine the three possible cases for $\mathscr{L}(\alpha)$:
\begin{description}
    \item[2.1] If $\mathscr{L}(\alpha)=\mathbf{in}$, then every attacker $\beta\in \alpha^{-}$ must satisfy $\mathscr{L}(\beta)=\mathbf{out}$. Hence, condition (ii) of a complete labelling holds for $\alpha$;
    \item[2.2] If $\mathscr{L}(\alpha)=\mathbf{out}$, then there exists $\beta\in \alpha^{-}$ with $\mathscr{L}(\beta)=\mathbf{in}$. Thus, condition (i) of a complete labelling holds for $\alpha$;
    \item[2.3] If $\mathscr{L}(\alpha)=\mathbf{und}$, then it is neither the case that all are $\mathbf{out}$ nor that some attacker is $\mathbf{in}$, which indicates that $\alpha$ must satisfy the condition (i) and (ii) of a complete labelling.
\end{description}
\end{description}

Since $\alpha$ was arbitrary, $\mathscr{L}$ satisfies the complete labelling conditions for every argument $\alpha \in \mathscr{A}$. With the additional arguments such as $\mathscr{K}_\alpha$ labeled $\mathbf{in}$, a complete labelling $\mathscr{L'}$ of modified AF $\mathscr{F}_\psi$ could be constructed from $\mathscr{L}$ with $\mathscr{L}(\alpha)=\mathscr{L}'(\alpha)$ holds for $\alpha \in \mathscr{A}$.

\end{proof}

\section*{Proof of Proposition~\ref{pro:forbut-for}}

\begin{proof}
    Suppose $\psi$ is a but-for cause for $\mathbf{v}(\alpha)$. By Def.~\ref{def:butforcause}, this implies that $(\mathscr{M}_{\mathscr{F}},\mathscr{L})\Vdash \psi\wedge \mathbf{v}(\alpha)$. Consequently, $\psi$ satisfies the Actuality ($\mathrm{AC1}$). Furthermore, the definition of but-for cause requires $\psi$ to be an atomic hypothesis. As an atomic hypothesis, $\psi$ possesses no proper sub-conjunctions. Therefore, $\psi$ satisfies the Minimality ($\mathrm{AC3}$).
\end{proof}

\section*{Proof of Proposition~\ref{pro:relevant}}

\begin{proof}
    We proceed by case analysis on the structure of the hypothesis $\psi$.
\begin{description}
    \item[Case 1] Assume $\psi$ is an atomic hypothesis, let $\psi=\mathbf{v'}(\beta)$.
    Suppose, for the sake of contradiction, that the set of arguments $\mathscr{A}_{\psi}$ contains at least one argument that is not relevant to $\alpha$. 
    Because $\psi$ is atomic, $\beta$ is the only element that is not relevant to $\alpha$. Since the label of argument is only affected by its parents' label, if changing the label of $\beta$ would affect the label of $\alpha$, then $\beta$ should be an ancestor of $\alpha$, which contradicts the fact that $\beta$ is not relevant to $\alpha$.
    
    \item[Case 2] Assume $\psi$ is a compound hypothesis. Suppose, for the sake of contradiction, that there exists $\beta\in \mathscr{A}_\psi$ s.t. $\beta$ is not relevant to $\alpha$. Consider the proper sub-conjunction $\psi'=\bigwedge_{\gamma\in\mathscr{A}_\psi\setminus \{\beta\}}\mathbf{v}(\gamma)$. If $\psi$ is an actual cause of $\mathbf{v}(\alpha)$, then $\psi$ satisfies $\mathrm{AC3}$, which suggests that the proper sub-conjunction $\psi'$ is not an actual cause of $\mathbf{v}(\alpha)$. Given that $\psi'$ is a proper sub-conjunction of $\psi$, if $\psi$ satisfies $\mathrm{AC1}$, then $\psi'$ satisfies $\mathrm{AC1}$.
    Consequently, the reason $\psi^{\prime}$ fails to be an actual cause is that it violates $\mathrm{AC2}$.
    Therefore, there exists no modified hypothesis of $\psi'$ that results in $\neg\mathbf{v}(\alpha)$. 
    However, because $\psi$ is an actual cause of $\mathbf{v}(\alpha)$, there exists a modified hypothesis of $\psi$ that causes $\neg\mathbf{v}(\alpha)$. The only difference between these two hypotheses is $\mathbf{v'}(\beta)$. 
    As established, changing the label of $\beta$ can only affect the label of $\alpha$ if $\beta$ is an ancestor of $\alpha$, which contradicts the assumption that $\beta$ is not relevant to $\alpha$.
\end{description}

\end{proof}

\section*{Proof of Proposition~\ref{pro:notMonotonicity}}

\begin{proof}
    We proceed by contradiction. Suppose $\psi'$ is an actual cause of $\mathbf{v}(\alpha)$, then there exists a proper sub-conjunction $\psi$ of $\psi'$ such that $\psi$ is again an actual cause of $\mathbf{v}(\alpha)$. This contradicts the fact that $\psi'$ satisfies $\mathrm{AC3}$.
\end{proof}

\section*{Proof of Proposition~\ref{pro:non-trivialexpl.}}
\begin{proof}
This proposition consists of two statements, which we prove respectively.
\vspace{0.5em}

\noindent\textbf{Proof of 1.}
This is proved by constructing a hypothesis $\psi$ that satisfies $\mathrm{AC1}$ and $\mathrm{AC2}^m_{\exists}$.

Given that $\alpha^- \neq \emptyset$ and $\mathscr{L}(\alpha)=\mathrm{in}$, for all $\gamma \in \alpha^-$, it holds that $\mathscr{L}(\gamma)=\mathrm{out}$. For each $\gamma \in \alpha^-$, there exists at least one parent $\eta \in \gamma^-$ s.t $\mathscr{L}(\eta)=\mathrm{in}$. Let $\gamma\in\alpha^-$ be a parent of $\alpha$, let $\mathscr{C}=\{ \eta~|~\mathscr{L}(\eta)=\mathrm{in}, \eta \in \gamma^- \}$ be the set and $\psi=\bigwedge_{\eta\in\mathscr{C}}\mathrm{in(\eta)}$ be the conjunction. This construction ensures that $\psi$ satisfies $\mathrm{AC1}$.

Then we prove that $\psi$ satisfies $\mathrm{AC2}_{\exists}^m$. Let $\psi'=\bigwedge_{\eta\in\mathscr{C}}\mathrm{out(\eta)}$. By Def.~\ref{def:modified-hyp}, $\psi'$ is a modified hypothesis of $\psi$. If the model is intervened by $\psi'$, then $\mathscr{L}(\gamma)\neq \mathrm{out}$ since no parent of $\gamma$ is labeled $\mathrm{in}$. This leads to the fact that $\mathscr{L}(\alpha)\neq \mathrm{in}$. So $\psi$ satisfies $\mathrm{AC2_{\exists}^m}$.

The constructed $\psi$ is a finite set since the model $\mathscr{M}$ is a finite model. By checking whether there exists a proper sub-conjunction of $\psi$ that satisfies $\mathrm{AC1}$ and $\mathrm{AC2}_{\exists}^m$. If not, then $\psi$ is the actual cause; if so, then again checking whether there exists a proper sub-conjunction. Within finite steps we could end with a $\psi'$ that is a proper sub-conjunction that satisfies $\mathrm{AC3}$, and the $\psi'$ is the actual cause.

\vspace{1.5em}
\noindent\textbf{Proof of 2.} The Overdetermination presented in Ex.~\ref{ex:overdetermination} serves as a counterexample.

\end{proof}

\section*{Proof of Proposition~\ref{pro:order0invariant}}
\begin{proof}
According to the restriction on hypotheses in Def.~\ref{def:syntax}, for any atomic formulas \(\mathbf{v_1}(\alpha_1), \mathbf{v_2}(\alpha_2) \in \mathsf{Sub}(\psi)\), we have \(\alpha_1 \neq \alpha_2\). Therefore, for any atomic formula \(\mathbf{v}(\alpha) \in \mathsf{Sub}(\psi)\), if \(\mathbf{v}(\alpha) \in \mathsf{Sub}(\psi_1)\), then there does not exist any \(\mathbf{v}' \in \mathbb{V}\) such that \(\mathbf{v}'(\alpha) \in \mathsf{Sub}(\psi_2)\). This shows that for any argument \(\alpha\), if a conjunct \(\psi_1\) of \(\psi\) intervenes on \(\alpha_i\), then there is no atomic formula related to \(\alpha_i\) in the other conjunct \(\psi_2\). (Denoted as Sub-conclusion 1.)

Based on this, we may assume that 
\(\psi = \mathbf{v_1}(\alpha_1) \land \mathbf{v_2}(\alpha_2) \land \cdots \land \mathbf{v_n}(\alpha_n)\) 
is an intervention consisting of \(n\) (\(n \ge 2\)) atomic formulas. After an arbitrary partition, the two conjuncts are denoted as 
\begin{align*}
    \psi_1 &= \mathbf{v_1}(\alpha_1) \land \mathbf{v_2}(\alpha_2) \land \cdots \land \mathbf{v_j}(\alpha_j),\\
\psi_2 &= \mathbf{v_{j+1}}(\alpha_{j+1}) \land \cdots \land \mathbf{v_n}(\alpha_n),
\end{align*}
where \(1 \le j < n\). For an acceptability model \(\mathscr{M}_{\mathscr{F}} = \{ f_\alpha \}_{\alpha\in\mathscr{A}}\), by Def.~\ref{def:Intervention}, we have 
\[
\mathscr{M}^{\psi}_{\mathscr{F}} = \{ f^{\psi}_\alpha \}_{\alpha\in \mathscr{A}},\quad 
f_{\alpha}^{\psi} = 
\begin{cases}
\mathbf{\mathbf{v_i}}, & \exists i\in \{1,\cdots,k\}: \alpha=\alpha_i; \\[3pt]
f_{\alpha}, & \text{otherwise}.
\end{cases}
\]

Next, we prove that \(\mathscr{M}^{\psi}_{\mathscr{F}} = (\mathscr{M}^{\psi_1}_{\mathscr{F}})^{\psi_2}=(\mathscr{M}^{\psi_2}_{\mathscr{F}})^{\psi_1}\), i.e., for any \(f_\alpha \in \mathscr{M}_{\mathscr{F}}\),
\[
f^{\psi}_{\alpha} = (f^{\psi_1}_{\alpha})^{\psi_2} = (f^{\psi_2}_{\alpha})^{\psi_1}.
\]

We proceed by case analysis on an arbitrary argument $\alpha\in \mathscr{A}$.

\begin{description}
    \item[Case 1] For an argument \(\alpha\) that does \emph{not} appear in \(\psi\), i.e., there is no \(\mathbf{v}\) such that \(\mathbf{v}(\alpha) \in \mathsf{Sub}(\psi)\). By Def.~\ref{def:Intervention}, \(f^{\psi}_{\alpha} = f_{\alpha}\). Similarly, since \(\psi_1\) and \(\psi_2\) are conjuncts of \(\psi\), from \(\mathbf{v}(\alpha) \notin \mathsf{Sub}(\psi)\) we obtain \(\mathbf{v}(\alpha) \notin \mathsf{Sub}(\psi_1)\) and \(\mathbf{v}(\alpha) \notin \mathsf{Sub}(\psi_2)\). By Def.~\ref{def:Intervention}, \((f^{\psi_1}_{\alpha})^{\psi_2} = f^{\psi_2}_{\alpha} = f_{\alpha}\). Likewise, \((f^{\psi_2}_{\alpha})^{\psi_1} = f_{\alpha}\).
    \item[Case 2] For an argument \(\alpha\) that appears in \(\psi_1\), i.e., there exists \(\mathbf{v}\) such that \(\mathbf{v}(\alpha) \in \mathsf{Sub}(\psi_1)\). By Def.~\ref{def:Intervention}, \(f^{\psi_1}_{\alpha} = \mathbf{v}\). By Sub-conclusion 1, there is no \(\mathbf{v}'\) such that \(\mathbf{v}'(\alpha) \in \mathsf{Sub}(\psi_2)\). Hence, \((f^{\psi_1}_{\alpha})^{\psi_2} = f^{\psi_1}_{\alpha} = \mathbf{v}\). Since \(\psi_1\) is a conjunct of \(\psi\) and \(\mathbf{v}(\alpha) \in \mathsf{Sub}(\psi_1)\), we have \(\mathbf{v}(\alpha) \in \mathsf{Sub}(\psi)\). Then, \(f^{\psi}_{\alpha} = \mathbf{v}\). Thus, \(f^{\psi}_{\alpha} = (f^{\psi_1}_{\alpha})^{\psi_2} = \mathbf{v}\). Similarly, \((f^{\psi_2}_{\alpha})^{\psi_1} = \mathbf{v}\).
    \item[Case 3] For an argument \(\alpha\) that appears in \(\psi_2\), i.e., there exists \(\mathbf{v}\) such that \(\mathbf{v}(\alpha) \in \mathsf{Sub}(\psi_2)\). By Sub-conclusion 1, there is no \(\mathbf{v}'\) such that \(\mathbf{v}'(\alpha) \in \mathsf{Sub}(\psi_1)\). By Def.~\ref{def:Intervention}, \(f^{\psi_1}_{\alpha} = f_{\alpha}\). Moreover, since \(\mathbf{v}(\alpha) \in \mathsf{Sub}(\psi_2)\), \((f^{\psi_1}_{\alpha})^{\psi_2} = f^{\psi_2}_{\alpha} = \mathbf{v}\). Because \(\psi_2\) is a conjunct of \(\psi\) and \(\mathbf{v}(\alpha) \in \mathsf{Sub}(\psi_2)\), we have \(\mathbf{v}(\alpha) \in \mathsf{Sub}(\psi)\). Then, \(f^{\psi}_{\alpha} = \mathbf{v}\). Thus, \(f^{\psi}_{\alpha} = (f^{\psi_1}_{\alpha})^{\psi_2} = \mathbf{v}\). Similarly, \((f^{\psi_2}_{\alpha})^{\psi_1} = \mathbf{v}\).
\end{description}

In summary, for any argument \(\alpha\), we have 
\[
f^{\psi}_{\alpha} = (f^{\psi_1}_{\alpha})^{\psi_2} = (f^{\psi_2}_{\alpha})^{\psi_1}.
\]

Given that \( \mathscr{M}^{\psi}_{\mathscr{F}}=\{ f^{\psi}_{\alpha} \}\), this results in 
\[
\mathscr{M}^{\psi}_{\mathscr{F}} = (\mathscr{M}^{\psi_1}_{\mathscr{F}})^{\psi_2}=(\mathscr{M}^{\psi_2}_{\mathscr{F}})^{\psi_1}.
\]
\end{proof}

\section*{Proof of Proposition~\ref{pro:counter-root}}

\begin{proof}
This proposition consists of two statements, which we prove respectively.
\vspace{0.5em}

\noindent\textbf{Proof of 1.}
Assume that $\psi\in \operatorname{Counter}(\mathscr{F},\alpha)$ with $\mathscr{A}_{\psi}\subseteq \mathbf{in}(\mathscr{L})_{I}\cup \mathbf{in}(\mathscr{L})_{S}$ and $(\alpha_i,\alpha)\in \overline{\mathbf{in}(\mathscr{L})}^{+}$ for all $\alpha_i\in \mathscr{A}_{\psi}$. There exists a complete labelling $\mathscr{L}\in \mathscr{L}_{co}(\mathscr{F})$ s.t.:
\begin{itemize}
    \item $\mathscr{L}(\alpha)=\mathbf{in}$;
    \item $\mathscr{A}_{\psi}\subseteq E_0$ with $E_0=\mathbf{in}(\mathscr{L})$;
    \item for every $\beta\in \mathscr{A}_{\psi}$, $\beta\in E_{0_I}\cup E_{0_S}$ and $(\beta,\alpha)\in \overline{E}^{+}$.
\end{itemize}

Now define:
\[
R_{E_0}:=\{\beta\in E_0\mid (\beta,\alpha)\in \overline{E}_{0}^{+}~\text{and}~\beta\in E_{0_I}\cup E_{0_S}\}, 
\]
which is exactly the root reason obtained from the complete extension $E_0$ according to the definition of Root Reasons. The condition above directly yields $\mathscr{A}_{\psi}\subseteq R_{E_0}$. Hence, the required root reason exists.

\vspace{0.5em}
\textbf{The converse fails (counterexample).}

 \begin{center}
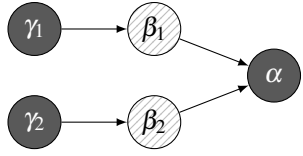

        \begin{minipage}[c]{0.4\textwidth}
            \centering
        \begin{tikzpicture}[scale=0.88,->,>=latex,
        auto,
        every node/.style={draw=black,circle,minimum size=0.7cm,inner sep=0pt},
    outnode/.style={text=black, pattern=north east lines, pattern color=gray!50},
    innode/.style={fill=gray!130, text=white},
    undnode/.style={fill=gray!50, text=white}
        ]
            \node[innode] (y1) at (-1.8,0.7) {$\gamma_{1}$};
            \node[innode] (y2) at (-1.8,-0.7) {$\gamma_{2}$};
            \node[outnode] (b1) at (0,0.7) {$\beta_{1}$};
            \node[outnode] (b2) at (0,-0.7) {$\beta_{2}$};
            \node[innode] (a) at (1.8,0) {$\alpha$};
        
            \draw (y1) -- (b1);
            \draw (b1) -- (a);
            \draw (y2) -- (b2);
            \draw (b2) -- (a);

        \end{tikzpicture}
            \captionof{figure}{$\mathscr{F}_6$: V-Shape}
            \label{fig:v-shape}
        \end{minipage}
        \hfill
        \begin{minipage}[c]{0.55\textwidth}
        \vspace{-10pt}
            Consider the AF $\mathscr{F}_6$ in Fig.~\ref{fig:v-shape} with the topic $\alpha$. The set $E=\{\gamma_1,\gamma_2\}$ is the only root reason for $\alpha$, while $\psi_1=\mathbf{in}(\gamma_1)$, $\psi_2=\mathbf{in}(\gamma_2)$ and $\psi_3=\mathbf{in}(\alpha)$ are the actual cause for $\mathbf{in}(\alpha)$. It is not the case that $E\subseteq \mathscr{A}_{\psi}$ for $i\in \{1,2,3\}$.
        \end{minipage}
        \vspace{-3pt}
    \end{center}

\vspace{1.5em}

\noindent\textbf{Proof of 2.}
Assume that $\psi\in\operatorname{Counter}(\mathscr{F},\alpha)$ with $\mathscr{A}_{\psi}\subseteq\mathbf{in}(\mathscr{L})$. There exists a complete labelling $\mathscr{L}\in \mathscr{L}_{co}(\mathscr{F})$ s.t., $\mathscr{L}(\alpha)=\mathbf{in}$ and for every $\alpha_i\in \mathscr{A}_{\psi}$, we have $\mathscr{L}(\alpha_i)=\mathbf{in}$. Let $E_0\in \mathscr{E}_{co}(\mathscr{F})$ with $\mathscr{A}_{\psi}\subseteq E_0$ and $\alpha\in E_0$, hence $E_0$ is admissible and conflict-free.

Define:
\[
S=\{\beta\in E_0\mid \beta~\text{is relevant to}~\alpha~\text{in}~\mathscr{F}\}.
\]
By Proposition~\ref{pro:relevant}, each argument in $\mathscr{A}_{\psi}$ is relevant to $\alpha$; therefore, $\mathscr{A}_{\psi}\subseteq S$.

We now show that $S$ is sufficient for $\alpha$.
\begin{itemize}
    \item \textbf{Relevance:} By construction, we know that every $\beta\in S$ is relevant for $\alpha$.
    \item \textbf{Conflict-freeness:} $S\subseteq E_0$ and $E_0$ is conflict-free, so $S$ is conflict-free.
    \item \textbf{Defence:} Take any $\beta\in S\cup \{\alpha\}$.
    \begin{itemize}
        \item If $\beta\in S$, then $\beta\in E_0$ and $E_0$ is admissible. For every attacker $\delta\in \beta^{-}$, there exists $\gamma\in E_0$ with $(\gamma,\delta)\in \mathscr{R}$. Because $\beta$ is relevant to $\alpha$, there exists a directed path from $\beta$ to $\alpha$. Extending this path via $(\delta,\beta)\in \mathscr{R}$ and $(\gamma,\delta)\in \mathscr{R}$ yields a directed path from $\gamma$ to $\alpha$. Thus $\gamma$ is relevant to $\alpha$ and consequently $\gamma\in S$.
        \item If $\beta=\alpha$, the same argument applies.
    \end{itemize}

\end{itemize}

Therefore, every attacker of $\beta$ is counter-attacked by an argument in $S$, i.e., $S$ defends $S\cup \{\alpha\}$.

Thus, $S$ satisfied the three conditions, the set $E=S\cup \{\alpha\}$ belongs to $\operatorname{Suf}(\mathscr{F},\alpha)$. Since $\mathscr{A}_{\psi}\subseteq S\subseteq E_{0}$, the required sufficient explanation exists.

\vspace{0.5em}
\textbf{Failure of the converse.}
Consider the AF in Fig.~\ref{fig:overdetermination} with the topic $\eta$. The set $\{\alpha,\eta\}\in \operatorname{Suf}(\mathscr{F},\eta)$; similarly, $\{\beta,\eta\}\in \operatorname{Suf}(\mathscr{F},\eta)$. Take the sufficient explanation $E=\{\alpha,\eta\}$. There exists no  $\psi\in\operatorname{Counter}(\mathscr{F},\alpha)$ with $\mathscr{A}_{\psi}\subseteq\mathbf{in}(\mathscr{L})$ and $\mathscr{A}_{\psi}\subseteq E\setminus \{\eta\}$, because for the candidates $\psi_1=\mathbf{in}(\alpha)\wedge\mathbf{in}(\beta)$ and $\psi_2=\mathbf{in}(\eta)$, it holds that $\mathscr{A}_{\psi_{i}}\not\subseteq \{\alpha\}$ for $i\in \{1,2\}$. Hence the inclusion $\mathscr{A}_{\psi}\subseteq E\setminus \{\eta\}$ is impossible for this $E$. This demonstrates that the converse of the statement does not hold in general.

\end{proof}

\end{document}